\documentclass[nohyperref]{article}

\usepackage{microtype}
\usepackage{graphicx}
\usepackage{subcaption}
\usepackage{booktabs} 

\usepackage{hyperref}

\usepackage[accepted]{icml2022_preprint}

\usepackage{amsmath}
\usepackage{amssymb}
\usepackage{mathtools}
\usepackage{amsthm}

\usepackage[capitalize,noabbrev]{cleveref}

\theoremstyle{plain}
\newtheorem{theorem}{Theorem}[section]
\newtheorem{proposition}[theorem]{Proposition}
\newtheorem{lemma}[theorem]{Lemma}
\newtheorem{corollary}[theorem]{Corollary}
\theoremstyle{definition}
\newtheorem{definition}[theorem]{Definition}
\newtheorem{assumption}[theorem]{Assumption}
\theoremstyle{remark}

\usepackage[textsize=tiny]{todonotes}

\usepackage{enumitem,amssymb}
\usepackage{soul}
\usepackage{bbm}
\usepackage{comment}
\usepackage{graphicx}
\usepackage{subcaption}
\usepackage{float}
\usepackage{algorithm} 
\usepackage{footnote}
\let\oldemptyset\emptyset
\DeclareMathOperator{\R}{\mathbb{R}}
\DeclareMathOperator{\E}{\mathbb{E}}

\newcommand{\Tau}{\mathcal{T}}
\newcommand\abs[1]{\left\lvert#1\right\rvert}
\newcommand\norm[1]{\left\lVert#1\right\rVert}

\DeclareMathOperator*{\argmax}{arg\,max}

\icmltitlerunning{Convergence and Price of Anarchy Guarantees of the Softmax Policy Gradient in Markov Potential Games}

\begin{document}

\twocolumn[
\icmltitle{Convergence and Price of Anarchy Guarantees of the \\ Softmax Policy Gradient in Markov Potential Games}



\icmlsetsymbol{equal}{*}

\begin{icmlauthorlist}
\icmlauthor{Dingyang Chen}{aiisc}
\icmlauthor{Qi Zhang}{aiisc}
\icmlauthor{Thinh T. Doan}{vt}
\end{icmlauthorlist}

\icmlaffiliation{aiisc}{Artificial Intelligence Institute, University of South Carolina}
\icmlaffiliation{vt}{Department of Electrical and Computer Engineering, Virginia Tech}

\icmlcorrespondingauthor{Dingyang Chen}{dingyang@email.sc.edu}
\icmlcorrespondingauthor{Qi Zhang}{qz5@cse.sc.edu}


\vskip 0.3in
]



\printAffiliationsAndNotice{}  

\begin{abstract}
We study the performance of policy gradient methods for the subclass of Markov games known as Markov potential games (MPGs), which extends the notion of normal-form potential games to the stateful setting and includes the important special case of the fully cooperative setting where the agents share an identical reward function. Our focus in this paper is to study the convergence of the policy gradient method for solving MPGs under softmax policy parameterization, both tabular and parameterized with general function approximators such as neural networks. We first show the asymptotic convergence of this method to a Nash equilibrium of MPGs for tabular softmax policies. Second, we derive the finite-time performance of the policy gradient in two settings: 1)  using the log-barrier regularization, and 2) using the natural policy gradient under the best-response dynamics (NPG-BR). Finally, extending the notion of price of anarchy (POA) and smoothness in normal-form games, we introduce the POA for MPGs and provide a POA bound for NPG-BR. To our knowledge, this is the first POA bound for solving MPGs. To support our theoretical results, we empirically compare the convergence rates and POA of policy gradient variants for both tabular and neural softmax policies.
\end{abstract}

\section{Introduction}
The framework of multi-agent sequential decision making is often formulated as (variants of) Markov games (MGs) \cite{shapley1953stochastic}, which finds a wide range of real-world applications such as coordination of multi-robot systems \citep{corke2005networked}, traffic control \citep{chu2019multi}, power grid management \citep{callaway2010achieving}, etc.
Perhaps the most well-known solution concept for MGs is the Nash policy, which is also known as the Nash equilibrium in the special case of stateless Markov games (i.e., normal-form games).
In a Nash policy, every agent selects its actions independently of any other agent given the state and plays a best response to all other agents.
In the special case of single-agent Markov games, aka Markov decision processes (MDPs), Nash policies reduce to the agent's optimal policies.
Most existing algorithms seeking to find Nash policies are value-based (i.e., computing only value functions related to the MG), with examples including Nash Q-learning \cite{hu2003nash}, Hyper-Q Learning \cite{tesauro2003extending}, and Nash-VI for the special case of zero-sum MGs \cite{zhang2020model}.
Policy-based algorithms, including multi-agent actor-critic algorithms, have recently gained attention with impressive empirical success \cite{lowe2017multi,foerster2017counterfactual} as well as provable guarantees \cite{zhang2018fully,leonardos2021global,zhang2021gradient}.

This paper focuses on the MG subclass of {\em Markov potential games} (MPGs) \cite{macua2018learning,leonardos2021global,zhang2021gradient}, which is extended from the notion of (normal-form) potential game and also incorporates as a special case the fully cooperative MGs where all agents share the same reward to optimize.
The MPG structure allows for exploiting recent advances in single-agent policy gradient methods (e.g., \cite{agarwal2019theory}) to establish the convergence of policy gradient to (near-)Nash policies in MPGs.
Specifically, existing work has established finite-time convergence guarantees under the {\em direct} policy parameterization. 
In this paper, we are interested in the alternative {\em softmax} policy parameterization, both tabularly and with neural networks for learnable state representations.
For tabular softmax, we establish several convergence guarantees to (near-)Nash policies in MPGs in Section \ref{sec:Convergence of the tabular softmax policy gradient in MPGs}, extending their counterpart from the single-agent setting \cite{agarwal2019theory}.
We then empirically compare tabular softmax with neural network-based softmax parameterization in terms of their convergence rates.

MPGs can model many problems where outcomes of high social welfare, measured by the sum of all agents' values, are most desirable.
In these scenarios, the solution concept of the Nash policy is inadequate.
The price of anarchy (POA) of a policy, firstly studied in normal-form games \cite{roughgarden2015intrinsic}, is accordingly defined as the ratio between the sum of all agents' value under this policy and the maximum-possible value sum.
In this sense, the POA further measures the quality of a Nash policy.   
In Section \ref{sec:Bounding the price of anarchy in smooth Markov (potential) games}, we  extend the notion of POA to the stateful MGs and provide first POA bounds for near-Nash policies in MGs and for an approximate best-response dynamics in MPGs.
We empirically compare the POA of Nash policies achieved by variants of softmax policy gradient dynamics.

\subsection{Related work}
\textbf{Single-agent policy gradient convergence.}
Agarwal et al. firstly established the policy gradient convergence of to global optima in the single-agent setting under tabular softmax parameterization, specifically, asymptotic convergence of policy gradient ascent, finite-time convergence with log barrier regularization, and finite-time convergence with natural policy gradient.
Agarwal et al. also established finite-time convergence for direct policy parameterization \cite{agarwal2019theory}.
Mei et al. later established finite-time convergence of (regularized) policy gradient ascent under tabular softmax parameterization, with a convergence rate depending on a problem-specific variable \cite{mei2020global}.
This problem-specific variable in some sense is necessary, as Li et al. have shown that softmax policy gradient can take exponential time to converge \cite{li2021softmax}.

\textbf{Policy gradient convergence in MPGs.}
Extending the work by Agarwal et al. \cite{agarwal2019theory} from the single-agent setting, Leonardos et al. \cite{leonardos2021global} and Zhang et al. \cite{zhang2021gradient} both established finite-time convergence of projected gradient ascent under tabular softmax parameterization to near-Nash policies in MPGs.   
Fox et al. \cite{fox2022independent} established the asymptotic convergence of natural policy gradient to Nash policies in MPGs. 

\textbf{POA bounds in normal-form games.}
Mirrokni and Vetta \cite{mirrokni2004convergence} initiated the discussion on the importance of POA bounds beyond Nash equilibria.
Roughgarden \cite{roughgarden2015intrinsic} defined the smoothness of (normal-form) games and then established the first POA bounds of on near-Nash equilibria in smooth games.
Roughgarden \cite{roughgarden2015intrinsic} provided POA bounds for the maximum-gain best-response dynamics in smooth (normal-form) potential games.


\section{Preliminaries}
\paragraph{Markov game.}
We consider a Markov game (MG) $\langle \mathcal{N},\mathcal{S},\mathcal{A}, P, \vec{r} \rangle$ with
$N$ agents indexed by $i\in\mathcal{N}=\{1,...,N\}$,
state space $\mathcal{S}$,
action space $\mathcal{A} = \mathcal{A}^1\times\cdots\times\mathcal{A}^N$,
transition function $P: \mathcal{S}\times\mathcal{A}\to\Delta(\mathcal{S})$,
reward functions $\vec{r}=\{r^i\}_{i\in\mathcal{N}}$ with $r^i: \mathcal{S}\times\mathcal{A}\to\R$ for each $i\in\mathcal{N}$,
and initial state distribution $\mu \in \Delta(\mathcal{S})$.
We assume full observability for simplicity, i.e., each agent observes the state $s\in\mathcal{S}$.
Under full observability, we consider {\em product policies}, $\pi:\mathcal{S}\to\times_{i\in\mathcal{N}}\Delta(\mathcal{A}^i)$, that is factored as the product of individual policies $\pi^i:\mathcal{S}\to\Delta(\mathcal{A}^i)$, $\pi(a|s) = \prod_{i\in\mathcal{N}}\pi^i(a^i|s)$.
Define the discounted return for agent $i$ from time step $t$ as $G^i_t = \sum_{l=0}^{\infty} \gamma^l r^i_{t+l}$, where $r^i_t:=r^i(s_t,a_t)$ is the reward at time step $t$ for agent $i$. For agent $i$, product policy $\pi=(\pi^1,...,\pi^N)$ induces a value function defined as $V^i_\pi(s_t) = \E_{s_{t+1:\infty}, a_{t:\infty}\sim\pi}[G^i_t|s_t]$, and action-value function $Q^i_\pi(s_t, a_t) = \E_{s_{t+1:\infty}, a_{t+1:\infty}\sim\pi}[G^i_t|s_t, a_t]$.
Following policy $\pi$, agent $i$'s cumulative reward starting from $s_0\sim\mu$ is denoted as $V^i_\pi(\mu):=\E_{s_0\sim\mu}[V^i_\pi(s_0)]$.

It will be useful to define the (unnormalized) {\em discounted state visitation measure} by following policy $\pi$ after starting at $s_0\sim\mu$:
\begin{align*}
    d^\pi_\mu(s):=\E_{s_0\sim\mu}\left[\sum_{t=0}^{\infty}\gamma^t{\rm Pr}^\pi(s_t=s|s_0)\right]
\end{align*}
where ${\rm Pr}^\pi(s_t=s|s_0)$ is the probability that $s_t=s$ after starting at state $s_0$ and following $\pi$ thereafter.
We make a standard assumption for the discounted state visitation distribution to be positive for every state under any policy, as formally stated in Assumption \ref{assumption:discounted state visitation distribution}.
\begin{assumption}
\label{assumption:discounted state visitation distribution}
For any $\pi$ and any state $s$ of the Markov game, $d^\pi_\mu(s) > 0$.
\end{assumption}

\paragraph{Markov potential game.}
\begin{definition}[Markov potential game]
\label{definition:Markov potential game}
A Markov game is called a {\em Markov potential game} (MPG) if there exists a potential function $\phi:\mathcal{S}\times\mathcal{A} \to \mathbb{R}$ such that for any agent $i$, any pair of product policies $(\pi^i,\pi^{-i}), (\bar{\pi}^{i},\pi^{-i})$, and any state $s$:
\begin{align*}
    &\E_{s_{t+1:\infty}, a_{t:\infty}\sim(\bar{\pi}^{i},\pi^{-i})}\left[\sum_{t=0}^{\infty} \gamma^t r^i(s_t,a_t)|s_0=s\right] \\
    &-\E_{s_{t+1:\infty}, a_{t:\infty}\sim(\pi^i,\pi^{-i})}\left[\sum_{t=0}^{\infty} \gamma^t r^i(s_t,a_t)|s_0=s\right]\\
    =& \E_{s_{t+1:\infty}, a_{t:\infty}\sim(\bar{\pi}^{i},\pi^{-i})}\left[\sum_{t=0}^{\infty} \gamma^t \phi(s_t,a_t)|s_0=s\right] \\
    &- \E_{s_{t+1:\infty}, a_{t:\infty}\sim(\pi^i,\pi^{-i})}\left[\sum_{t=0}^{\infty} \gamma^t \phi(s_t,a_t)|s_0=s\right]
    .
\end{align*}
\end{definition}
Given a product policy $\pi$, we define the {\em total potential function} as $\Phi_\pi(s) := \E_{s_{t+1:\infty}, a_{t:\infty}\sim\pi}\left[\sum_{t=0}^{\infty} \gamma^t \phi(s_t,a_t)|s_0=s\right]$, and we can obtain that, for any agent $i$,
\begin{align}\label{eq:total potential function}
V^i_{\bar{\pi}^{i},\pi^{-i}}(s) -  V^i_{\pi^i,\pi^{-i}} (s)  =& \Phi_{\bar{\pi}^{i},\pi^{-i}}(s) -  \Phi_{\pi^i,\pi^{-i}}(s) \\
\text{giving}\quad 
\nabla_{\theta^i} V^i_\theta(s) =& \nabla_{\theta^i} \Phi_\theta(s). \nonumber
\end{align}
We also similarly define $\Phi_\pi(\mu) := \E_{s_0\sim\mu}[\Phi_\pi(s_0)]$.

As formally stated in Assumption \ref{assumption:Potential function is bounded}, we assume that $\phi$, and therefore $\Phi$, are bounded.
\begin{assumption} [Potential function is bounded]
\label{assumption:Potential function is bounded}
The potential function $\phi$ is bounded, such that the total potential function $\Phi$ is bounded as  $\Phi_{\rm min}\leq\Phi_\pi(s)\leq\Phi_{\rm max}~\forall s, \pi$.
\end{assumption}

\paragraph{Nash policy.}
We focus on the solution concept of ($\epsilon$-)Nash policy, as formally defined below.
\begin{definition}[$\epsilon$-Nash policy]
\label{definition:epsilon-Nash policy}
The {\em Nash-gap} of a policy $\pi$ is defined as \begin{align*}
    \mbox{Nash-gap}(\pi) := \max_{i}\left(\max_{\bar{\pi}^{i}} V^i_{\bar{\pi}^{i},\pi^{-i}}(\mu) - V^i_{\pi}(\mu)\right) 
\end{align*}
A product policy $\pi=(\pi_1,...,\pi_N)$ is an {\em $\epsilon$-Nash policy} if  $\mbox{Nash-gap}(\pi)\leq\epsilon$.
\end{definition}


\section{Convergence of the tabular softmax policy gradient in MPGs}
\label{sec:Convergence of the tabular softmax policy gradient in MPGs}
In this section, we consider individual policies $(\pi_1,...,\pi_N)$ to be independently parameterized in the softmax tabular manner from the global state, i.e., we have, for each agent $i$, its policy parameter $\theta^i=\{\theta^i_{s,a^i}: s\in\mathcal{S}, a^i\in\mathcal{A}^i\}$ and policy 
\begin{align*}
\pi^i_{\theta^i}(a^i|s)=\frac{\exp(\theta^i_{s,a^i})}{\sum_{\bar{a}^i\in\mathcal{A}^i}\exp(\theta^i_{s,\bar{a}^i})}.
\end{align*}
For the rest of this paper, we will abbreviate $\Phi_{\pi_\theta}$, $V^i_{\pi_\theta}$, $A^i_{\pi_\theta}$ as
$\Phi_{\theta}$, $V^i_{\theta}$, $A^i_{\theta}$, respectively.
Lemmas \ref{lemma:state-based tabular softmax multi-agent policy gradient} and \ref{lemma:smoothness_tabular_softmax} formally states the policy gradient form and the smoothness under the tabular softmax parameterization, respectively, which will be used to establish the convergence results in this section.

\begin{lemma}[Multi-agent tabular softmax policy gradient form, proof in Appendix \ref{sec:Proof of Lemma lemma:state-based tabular softmax multi-agent policy gradient}]
\label{lemma:state-based tabular softmax multi-agent policy gradient}
For the state-based tabular softmax multi-agent policy parameterization, we have:
\begin{align}\label{eq:state-based tabular softmax multi-agent policy gradient}
    \frac{\partial \Phi_\theta(\mu)}{\partial \theta^i_{s,a^i}} = \frac{\partial V^i_\theta(\mu)}{\partial \theta^i_{s,a^i}} = d^{\pi_\theta}_\mu(s) \pi^i_{\theta^i}(a^i|s)A^i_\theta(s, a^i)
\end{align}
where $A^i_\theta(s, a^i):=\E_{a^{-i}\sim\pi^{-i}_{\theta^{-i}}(\cdot|s)}[A^i_\theta(s, a^i, a^{-i})]$.
\end{lemma}

\begin{lemma}[Smoothness of $\Phi$ under tabular softmax, proof in Appendix \ref{sec:Proof of Lemma lemma:smoothness_tabular_softmax}]
\label{lemma:smoothness_tabular_softmax}
Under tabular softmax $\pi_\theta$, $\Phi_\theta(s)$ is $\frac{41N}{4(1-\gamma)^3}$-smooth for any state $s$ (hence for any initial state distribution $\mu$).
\end{lemma}

We next present our convergence results for the standard policy gradient dynamics without and with log barrier regularization in Sections \ref{sec:Asymptotic convergence of the policy gradient dynamics} and \ref{sec:Policy gradient dynamics with log-barrier regularization}, respectively, where Assumptions \ref{assumption:discounted state visitation distribution} and \ref{assumption:Potential function is bounded} hold.

\subsection{Asymptotic convergence of the policy gradient dynamics}
\label{sec:Asymptotic convergence of the policy gradient dynamics}
In Theorem \ref{theorem:Asymptotic convergence to Nash with gradient ascent}, we establish, under the tabular softmax policy parameterization, the asymptotic convergence to a Nash policy in a MPG of the standard policy gradient dynamics:
\begin{align}\label{eq:PG}
    \theta^i_{t+1} = \theta^i_{t} + \eta\nabla_{\theta^i} V^i_{\theta_{t}}(\mu) 
    = \theta^i_{t} + \eta\nabla_{\theta^i} \Phi_{\theta_{t}}(\mu) 
\end{align}
where $\eta$ is the fixed stepsize and the update is performed by every agent $i\in\mathcal{N}$. Theorem \ref{theorem:Asymptotic convergence to Nash with gradient ascent} relies on the assumption on the asymptotic convergence of the policy parameters, formally stated as follows.
\begin{assumption} 
\label{assumption:PG}
Following the policy gradient dynamics \eqref{eq:PG}, the policy parameter of every agent $i$ converges asymptotically, i.e., $\theta^i_t\to\theta^i_*$ as $t\to\infty,~\forall i$.
\end{assumption}

We remark here that the assumption that $\theta^i$ converges is made to ensure the convergence of $\{Q^i(s,a^i)\}_i$, which is then used to prove the theorem in a similar manner to \cite{agarwal2019theory}.
Note that, since the gradient is as Equation \eqref{eq:state-based tabular softmax multi-agent policy gradient}, the gradient converging to zero cannot directly imply the parameters converging to zero.
A sufficient condition for Assumption \ref{assumption:PG} to hold is that the stationary points of are {\em isolated}, which is originally assumed in Fox et al. \cite{fox2022independent} to establish the asymptotic convergence of natural policy gradient to Nash policies. 

\begin{theorem}[Asymptotic convergence of policy gradient, proof in Appendix \ref{sec:Proof of Theorem Asymptotic convergence to Nash with gradient ascent}]
\label{theorem:Asymptotic convergence to Nash with gradient ascent}
Suppose every agent $i\in\mathcal{N}$ follows the policy gradient dynamics \eqref{eq:PG} with $\eta\leq \text{min}(\frac{1-\gamma}{N\text{max}(5,\sqrt{N})(\Phi_{\text{max}}-\Phi_{\text{min}})},\frac{4(1-\gamma)^3}{41N})$ and Assumption \ref{assumption:PG} holds such that $\theta^i_t\to\theta^i_*$ for every agent $i$, then the product policy defined by $\theta_*=\{\theta^i_*\}_{i\in\mathcal{N}}$ is a Nash policy.
\end{theorem}

\subsection{Policy gradient dynamics with log-barrier regularization}
\label{sec:Policy gradient dynamics with log-barrier regularization}
Inspired by \cite{agarwal2019theory} for the single-agent setting, we consider the log barrier regularized objective as defined below to establish finite-time convergence guarantees for the policy gradient dynamics:
\begin{align*}
    L_\lambda(\theta) :=& \Phi_\theta(\mu) -\lambda\textstyle\sum_{i=1}^N \E_{s\sim {\rm Unif}_\mathcal{S}}\left[{\rm KL}({\rm Unif}_{\mathcal{A}^i}, \pi_\theta(\cdot|s))\right] \\
    =& \Phi_\theta(\mu) +
    \lambda\textstyle\sum_{i=1}^N\left(\frac{\sum_{s,a^i}\log\pi^i_{\theta^i}(a^i|s)}{|\mathcal{S}||\mathcal{A}^i|}+\log{|\mathcal{A}^i|}\right)
\end{align*}
where the log barrier regularization, i.e., the KL divergence with respect to the uniform action-selection distribution, is applied to each agent's policy independently.
Lemma \ref{lemma:Log barrier regularization's approximate first-order stationary points are near-Nash} extends the results in \cite{agarwal2019theory} to the multi-agent setting, stating that, with the log barrier regularization, approximate first-order stationary points are near-Nash.
\begin{lemma}
[Log barrier regularization's approximate first-order stationary points are near-Nash, proof in Appendix \ref{sec:Proof of Lemma lemma:Log barrier regularization's approximate first-order stationary points are near-Nash}]
\label{lemma:Log barrier regularization's approximate first-order stationary points are near-Nash}
Suppose $\theta$ is such that $\norm{\nabla_\theta L_\lambda(\theta)}_2\leq\lambda/(2|\mathcal{S}|\max_i|\mathcal{A}^i|)$. Then the product policy $\pi_\theta=(\pi^1_{\theta^1},...,\pi^N_{\theta^N})$ is a $2\lambda M$-Nash policy where 
$M:= \max_{\pi, \pi'}\norm{\frac{d^{\pi}_\mu}{d^{\pi'}_\mu}}_{\infty}$, which is well-defined by Assumption \ref{assumption:discounted state visitation distribution}.
\end{lemma}

With Lemma \ref{lemma:Log barrier regularization's approximate first-order stationary points are near-Nash}, we establish the convergence rate as stated in Theorem \ref{theorem:Convergence rate with log barrier regularization}.
\begin{theorem}
[Convergence rate of the policy gradient with log barrier regularization, proof in Appendix \ref{sec:Proof of Theorem theorem:Convergence rate with log barrier regularization}]
\label{theorem:Convergence rate with log barrier regularization}
Letting $\beta_\lambda:=\frac{41N}{4(1-\gamma)^3}+\frac{2\lambda N}{|\mathcal{S}|}$, then $\beta_\lambda$ is an upper bound on the smoothness of $L_\lambda(\theta)$. 
Starting from $\theta_0 = 0$, consider the updates $\theta_{t+1} = \theta_t + \eta \nabla_\theta L_\lambda(\theta_t)$ with $\lambda = \epsilon/2M $ and $\eta=1/\beta_\lambda$. Then, for any initial distribution $\mu$,
we have $\min_{t<T} \mbox{Nash-gap}_t \leq \epsilon $ whenever
\begin{align*}
    T\geq&\frac{328NM^2|\mathcal{S}|^2\max_i|\mathcal{A}^i|^2(\Phi_{\rm max}-\Phi_{\rm min})}{(1-\gamma)^3\epsilon^2} 
    \\ & + \frac{32 NM|\mathcal{S}|\max_i|\mathcal{A}^i|^2(\Phi_{\rm max}-\Phi_{\rm min})}{\epsilon}.
\end{align*}
\end{theorem}

\subsection{Approximate best-response natural policy gradient dynamics}
\label{sec:Approximate best-response natural policy gradient dynamics}
In this subsection, we consider the natural policy gradient (NPG) dynamics extended from the single-agent setting to Markov potential games.
The NPG dynamics is defined as 
\begin{align}\label{eq:NPG}
    \theta^i_{t+1} = \theta^i_{t} + \eta (F^i_{\theta_t})^\dagger \nabla_{\theta^i} V^i_{\theta_{t}}(\mu) 
    = \theta^i_{t} + \eta(F^i_{\theta_t})^\dagger\nabla_{\theta^i} \Phi_{\theta_{t}}(\mu),
\end{align}
where $A^\dagger$ denotes the Moore–Penrose inverse of a matrix $A$ and $F^i_{\theta}$ is the Fisher information matrix for agent $i$ under product policy $\pi_\theta$:
\begin{align*}
    F^i_{\theta} = 
    \E_{s\sim d^{\pi_\theta}_\mu, a^i \sim \pi^i(s)} \left[\nabla_{\theta^i}\log\pi^i_{\theta^i}(a^i|s) \nabla_{\theta^i}\log\pi^i_{\theta^i}(a^i|s)^\top \right].
\end{align*}

\begin{lemma}[NPG is effectively soft policy iteration, proof in Appendix \ref{sec:Proof of Lemma lemma:NPG is effectively soft policy iteration}]
\label{lemma:NPG is effectively soft policy iteration}
For any agent $i$, the NPG update \eqref{eq:NPG} is effectively:
\begin{align*}
    \theta^i_{t+1} =& \theta^i_{t} + \eta A^i_{\theta_t} 
    \quad \text{ and } \\
    \pi^i_{\theta^i_{t+1}}(a^i|s) =& \pi^i_{\theta^i_{t}}(a^i|s) \frac{\exp\left(\eta A^i_{\theta_t}(s,a^i)\right)} {Z^i_t(s)}
\end{align*}
where $Z^i_t(s)=\sum_{a^i}\pi^i_{\theta^i_{t}}(a^i|s)\exp\left(\eta A^i_{\theta_t}(s,a^i)\right)$ is the normalization constant for the softmax.
\end{lemma}

In the single-agent setting with the tabular softmax parameterization, we know that the NPG update (soft policy iteration) can achieve $O(1/\epsilon)$ convergence rate with $\epsilon$ being the single-agent optimality gap, compared with the  $O(1/{\epsilon^2})$ convergence rates achieved by (projected) gradient ascent methods for the direct parameterization and for the tabular softmax parameterization with the log barrier regularization \cite{agarwal2019theory}.
For MPGs, Fox et al. \cite{fox2022independent} established the asymptotic convergence of natural policy gradient. However, deriving the finite-time convergence with the tabular softmax parameterization when the agents concurrently perform the soft policy iteration is challenging, primarily due to the technical difficulty of relating the potential function value and the Nash-gap.
Here, we take a step back and consider the non-concurrent soft policy iteration, where an agent will perform a number of soft policy iterations with fixing other agents' policies:
letting $\theta^i_{t, 0} = \theta^i_t$, for $k=1,...,K$:
\begin{align}\label{eq:NPG-BR inner}
    \theta^i_{t, k} =& \theta^i_{t, k} + A^i_{t, k-1} 
    \text{ with }\\ 
    A^i_{t, k-1} (s,a^i) =& \E_{a^{-i}\sim\pi^{-i}_{\theta^{-i}_{t}}(\cdot|s)}\left[A^i_{\theta^i_{t, k-1},\theta^{-i}_{t}}(s,a^i,a^{-i})\right]\nonumber,
\end{align}
where $A^i_{t, k}$ is agent $i$'s local advantage of its policy currently parameterized by $\theta^{-i}_t$ with respect to the other agents' policies parameterized by $\theta^{-i}_t$, and $K$ is a hyperparameter that controls how close agent $i$ will get to its best response to $\theta^{-i}_t$.

The above update is performed independently for all agents,
and for the next iteration $t+1$ we only keep the change of the agent that induces the maximum gain in its own value and, equivalently, in the total potential function:
\begin{align}\label{eq:NPG-BR outer}
    i^*_t =& \argmax_i \Phi_{\theta^i_{t,K},\theta^{-i}_t}(\mu)-\Phi_{\theta_t}(\mu), \nonumber\\
    \theta^{i^*_t}_{t+1} =& \theta^{i^*_t}_{t,K} \quad\text{and}\quad \theta^{i}_{t+1} = \theta^{i}_{t} \text{ for } i\neq i^*_t
\end{align}
which ensembles the standard maximum-gain best-response dynamics for normal-form games \cite{roughgarden2016twenty}.

Suppose we aim to converge to a $\epsilon$-Nash policy.
We can set $K$ large enough (specifically $K\geq\frac{4}{(1-\gamma)^2 \epsilon}$ \cite{agarwal2019theory}), such that every agent's inner-loop update (indexed by $k$ \eqref{eq:NPG-BR inner}) achieves at least $\epsilon/2$-near-best-response.
Therefore, if no agent's improvement in their local value or, equivalently, in the total potential function as computed in \eqref{eq:NPG-BR outer} is no larger than $\epsilon/2$, then the product policy is already a $\epsilon$-Nash policy; otherwise, we can significantly improve the total potential function such that the total number of outer-loop updates, indexed by $t$ in \eqref{eq:NPG-BR outer}, can be bounded.
This establishes the convergence rate of our approximate-best-response NPG dynamics (\ref{eq:NPG-BR inner},\ref{eq:NPG-BR outer}), as formally stated in Theorem \ref{theorem:Convergence of approximate-best-response NPG}.

\begin{theorem}[Convergence of the approximate-best-response NPG, proof in Appendix \ref{sec:Proof of Theorem theorem:Convergence of approximate-best-response NPG}]
\label{theorem:Convergence of approximate-best-response NPG}
Setting $K\geq\frac{4}{(1-\gamma)^2 \epsilon}$ for as the iteration complexity of the inner-loop \eqref{eq:NPG-BR inner}, then the approximate-best-response NPG dynamics (\ref{eq:NPG-BR inner},\ref{eq:NPG-BR outer}) converges to a $\epsilon$-Nash policy within 
$O(\frac{\Phi_{\rm max} - \Phi_{\rm min}}{(1-\gamma)^2 \epsilon^2})$ inner-loop steps.
\end{theorem}

\section{Bounding the price of anarchy in smooth Markov (potential) games}
\label{sec:Bounding the price of anarchy in smooth Markov (potential) games}
In Definition \ref{definition:Price of anarchy in Markov games}, we formally define the price of anarchy in Markov games, which directly extends the notion in normal-from games that measures the quality of a product policy in terms of maximizing the sum of all agents' values.

\begin{definition}[Price of anarchy in Markov games]
\label{definition:Price of anarchy in Markov games}
The {\em price of anarchy} (POA) of a product policy $\pi$ is defined as $\frac{\sum_i V^i_\pi(\mu)}{\max_{\bar{\pi}}\sum_i V^i_{\bar{\pi}}(\mu)}$, i.e., the ratio between the values summed over all agents and the largest summed values achieved by any product policy $\bar{\pi}$.
\end{definition}
For the rest of this section, we formally extend the notion of smoothness from normal-form games \cite{roughgarden2015intrinsic} to Markov games in Section \ref{sec:Definition and sufficient conditions of smooth Markov game}, and present our POA bounds in smooth Markov (potential) games in Sections \ref{sec:POA bound for near-Nash policies} and \ref{sec:POA bound for the approximate best-response dynamics}.

\subsection{Definition and sufficient conditions of smooth Markov game.}
\label{sec:Definition and sufficient conditions of smooth Markov game}
Definition \ref{definition:Smooth Markov game} extends the notion of smooth normal-form game to its counterpart in Markov games.
\begin{definition}[Smooth Markov game]
\label{definition:Smooth Markov game}
A Markov game is $(\alpha,\beta)$-smooth if
\begin{align*}
    \textstyle\sum_i V^i_{\pi^i_\prime,\pi^{-i}}(s) \geq \alpha V_{\pi_\prime}(s)  - \beta V_{\pi}(s) 
\end{align*}
for any $s$ and any pair of product policies $\pi,\pi_\prime$, where $V_\pi(s):=\sum_i V^i_\pi(s)$.
\end{definition}

Intuitively, in a smooth Markov game, the externality imposed by one agent on the value of the others is limited.
Therefore, we conjecture that a sufficient condition is that both the transition and reward functions of the Markov game are ``smooth''. Proposition \ref{proposition:Transition and reward smoothness as a sufficient condition for Markov game smoothness} verifies this conjecture, which formally defines the smoothness of the transition and reward functions and establishes it as a sufficient condition for the smoothness of the Markov game.
\begin{proposition}[Transition and reward smoothness as a sufficient condition for Markov game smoothness]
\label{proposition:Transition and reward smoothness as a sufficient condition for Markov game smoothness}
The reward functions $\{r^i\}_{i\in\mathcal{N}}$ of a Markov game is said to be {\em $(\lambda,\mu)$-smooth} if 
\begin{align*}
    \lambda r_{\pi_\prime}(s) \leq \textstyle\sum_i r^i_{\pi^i_\prime, \pi^{-i}}(s) \leq \mu r_{\pi}(s)
\end{align*}
for any state $s$ and any pair of product policies $\pi, \pi_\prime$, where $r^i_{\pi}(s):=\E_{a\sim\pi(s)}[r^i(s,a)]$ and $r_\pi(s):=\sum_i r^i_\pi(s)$. 
Letting $M_\pi := (I-\gamma P_\pi)^{-1}$, the transition function $P$ of a Markov game is said to be  {\em $(\kappa,\nu)$-smooth} if
\begin{align*}
    M_{\pi^i_\prime, \pi^{-i}} r \geq \kappa M_{\pi_\prime}r - \nu M_{\pi}r
\end{align*}
for any $r\in\mathbb{R}^{|\mathcal{S}|}$ and any pair of product policies $\pi, \pi_\prime$.
For a Markov game, if its reward functions are $(\lambda,\mu)$-smooth and its transition function is $(\kappa,\nu)$-smooth, then the Markov game is $(\alpha=\kappa\lambda,\beta=\mu\nu)$-smooth.
\end{proposition}
\begin{proof}
We can establish
\begin{align*}
  \textstyle\sum_i V^i_{\pi^i_\prime,\pi^{-i}}(s) =& \textstyle\sum_i M_{\pi^i_\prime,\pi^{-i}} r^i_{\pi^i_\prime,\pi^{-i}} \\
  \geq& \textstyle\sum_i \kappa M_{\pi_\prime}r^i_{\pi^i_\prime,\pi^{-i}} - \nu M_{\pi}r^i_{\pi^i_\prime,\pi^{-i}} \\
  =& \kappa M_{\pi_\prime} \textstyle\sum_i r^i_{\pi^i_\prime,\pi^{-i}} - \mu M_{\pi} \textstyle\sum_i r^i_{\pi^i_\prime,\pi^{-i}}\\
  \geq& \kappa M_{\pi_\prime}\lambda r_{\pi_\prime} - \mu M_{\pi} \mu r_{\pi}  = \kappa\lambda V_{\pi_\prime} - \mu\nu V_{\pi}
\end{align*}
where the two inequalities are due to the smoothness of the transition function and the reward functions, respectively, which completes the proof.
\end{proof}

\subsection{POA bound for near-Nash policies}
\label{sec:POA bound for near-Nash policies}
We here derive our POA bound in Theorem \ref{theorem:POA of Nash in smooth Markov games} for near-Nash policies in smooth Markov games, generalizing from smooth normal-form games \cite{roughgarden2016twenty} to smooth Markov games.
Similar to the normal-form game counterpart, we describe the result for {\em $\epsilon$-ratio}-Nash policies as defined in Definition \ref{definition:epsilon-ratio-Nash policy} to ease presentation.

\begin{definition}[$\epsilon$-ratio-Nash policy]
\label{definition:epsilon-ratio-Nash policy}
A product policy $\pi=(\pi_1,...,\pi_N)$ is an {\em $\epsilon$-ratio-Nash policy} if, for any agent $i$,
$\max_{\pi^i_\prime} V^i_{\pi^i_\prime,\pi^{-i}}(\mu)\leq(1-\epsilon) V^i_{\pi}(\mu)$.
\end{definition}
\begin{theorem}[POA of $\epsilon$-ratio-Nash in smooth Markov games]
\label{theorem:POA of Nash in smooth Markov games}
In any $(\alpha,\beta)$-smooth Markov game, the POA of any $\epsilon$-ratio-Nash policy is at least $\frac{(1-\epsilon)\alpha}{1+(1-\epsilon)\beta}$.  
\end{theorem}
\begin{proof}
Consider setting $\pi_\prime=\pi_*$ in Definition \ref{definition:Smooth Markov game} where $\pi_*$ is a policy that achieves the optimal joint value, we have 
\begin{align*}
    V_{\pi}(s) = \textstyle\sum_i V^i_{\pi}(s) \geq& \textstyle\sum_i (1-\epsilon) V^i_{\pi^i_*,\pi^{-i}}(s) \\
    \geq& (1-\epsilon)\left(\alpha V_{\pi_*}(s)  - \beta V_{\pi}(s) \right)
\end{align*}
where the first inequality is by the definition of $\pi$ being $\epsilon$-ratio-Nash and the second inequality is by the definition of smooth Markov game.
Rearranging the terms completes the proof.
\end{proof}

\subsection{POA bound for the approximate best-response dynamics}
\label{sec:POA bound for the approximate best-response dynamics}
Inspired by the POA bounds for the best-response dynamics in smooth (normal-form) potential games \cite{roughgarden2015intrinsic}, we here derive the counterpart for smooth MPGs in Theorem \ref{theorem:POA bound of maximum-gain epsilon-ratio-best-response in smooth MPGs}, which bounds the number of policies generated from the {\em maximum-gain $\epsilon$-ratio-best-response} dynamics: Until product policy $\pi$ is $\epsilon$-ratio-Nash, update the maximum-gain agent to its best response, where the maximum-gain agent is the agent that induces the maximum increase in value after its best response. 

\begin{figure*}[ht]
    \centering
    \begin{subfigure}{\textwidth}
    \centering
    \includegraphics[width=.91\textwidth]{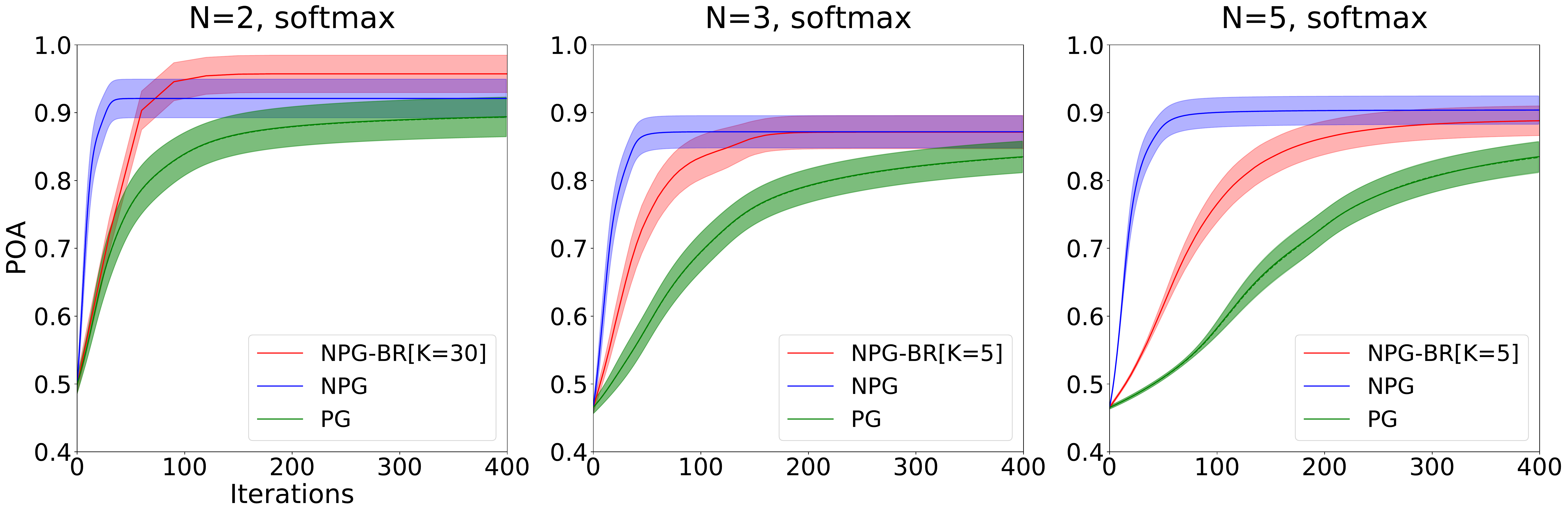}
    \end{subfigure}
    \begin{subfigure}{\textwidth}
    \centering
    \includegraphics[width=.91\textwidth]{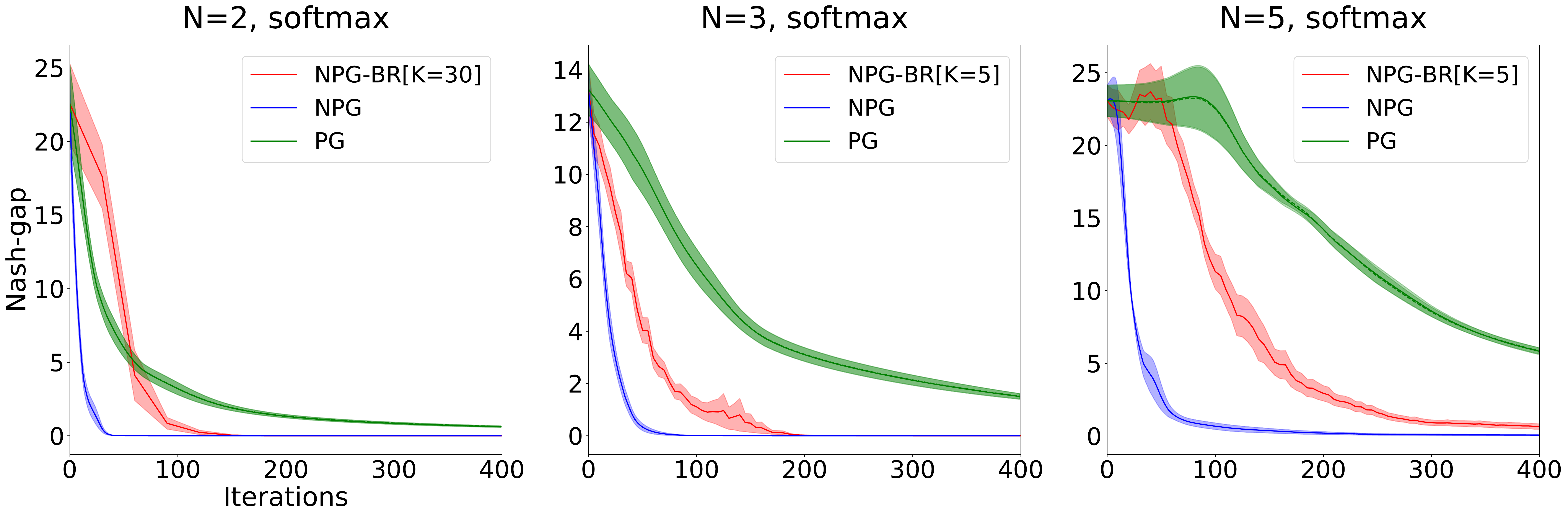}
    \end{subfigure}
    \caption{POA (top) and Nash-gap (bottom) under the tabular softmax parameterization (means and standard errors over 10 random initializations).The dashed lines are the curves of the log barrier regularized version of the algorithms with the same color. }
    \label{fig:Softmax-POA-Nash-gap}
\end{figure*}

We make Assumption \ref{assumption:potential leq V} that also appears in the normal-form game setting \cite{roughgarden2015intrinsic}. 
\begin{assumption}
\label{assumption:potential leq V}
We have $0<\Phi_\pi(s) \leq V_\pi(s)$ for any product policy $\pi$ and any state $s$.
\end{assumption}

\begin{theorem}[POA bound of maximum-gain $\epsilon$-ratio-best-response in smooth MPGs, proof in Appendix \ref{sec:Proof of Theorem theorem:POA bound of maximum-gain epsilon-ratio-best-response in smooth MPGs}]
\label{theorem:POA bound of maximum-gain epsilon-ratio-best-response in smooth MPGs}
Consider a $(\alpha,\beta)$-smooth MPG where Assumption \ref{assumption:potential leq V} holds. Let $\pi_*=\argmax_{\pi}V_\pi(\mu)$ be a globally optimal policy and $\sigma>0$ be a constant for analysis.
Consider the sequence of maximum-gain $\epsilon$-ratio-best-response policies $\pi_0, ...,\pi_T$.
Then, all but at most
\begin{align} \label{eq:bound the number of bad, epsilon-ratio}
    \log_{\rho} {\frac{\Phi_{\rm max}}{\Phi_0}} - T \log_{\rho} {\frac{1}{1-\epsilon}}
\end{align}
policies $\pi_t$ in the sequence satisfy
\begin{align}\label{eq:good policies}
    V_{\pi_t}(\mu) \geq \frac{\alpha}{(1+\beta)(1+\sigma)} V_{\pi_*}(\mu)
\end{align}
where $\rho = (1-\epsilon)(1+\frac{\sigma(1+\beta)}{N})$ and $\Phi_t := \Phi_{\pi_t}(\mu)$.
\end{theorem}

Since our NPG dynamics (\ref{eq:NPG-BR inner},\ref{eq:NPG-BR outer}) described in Section \ref{sec:Approximate best-response natural policy gradient dynamics}  is an instance of maximum-gain approximate-best-response dynamics, we have Corollary \ref{corollary:Bounding the POA of NPG-BR dynamics in smooth MPGs} directly induced by Theorem \ref{theorem:POA bound of maximum-gain epsilon-ratio-best-response in smooth MPGs}.

\begin{corollary}[POA bound of the approximate-best-response NPG dynamics (\ref{eq:NPG-BR inner},\ref{eq:NPG-BR outer}) in smooth MPGs, proof in Appendix \ref{sec:Proof of Corollary corollary:Bounding the POA of NPG-BR dynamics in smooth MPGs}]
\label{corollary:Bounding the POA of NPG-BR dynamics in smooth MPGs}
Consider a $(\alpha,\beta)$-smooth MPG where Assumption \ref{assumption:potential leq V} holds. Let $\pi_*=\argmax_{\pi}V_\pi(\mu)$ be a globally optimal policy and $\sigma>0$ be a constant for analysis.
Consider the sequence of policies $\pi_0, ...,\pi_T$ generated from the approximate-best-response NPG dynamics (\ref{eq:NPG-BR inner},\ref{eq:NPG-BR outer}) with $K\geq\frac{4}{(1-\gamma)^2 \epsilon}$.
Then, all but at most
\begin{align} \label{eq:bound the number of bad, NPG-BR}
    \log_{\rho} {\frac{\Phi_{\rm max}}{\Phi_0}} - T \log_{\rho} {\left(1+\frac{\epsilon}{2(1-\gamma)}\right)}
\end{align}
policies $\pi_t$ in the sequence satisfy \eqref{eq:good policies}, where 
$\rho = (1+\frac{\sigma(1+\beta)}{N})/(1+\frac{\epsilon}{2(1-\gamma)})$.
\end{corollary}
\
\begin{figure*}[ht]
    \centering
    \begin{subfigure}{\textwidth}
    \centering
    \includegraphics[width=.91\textwidth]{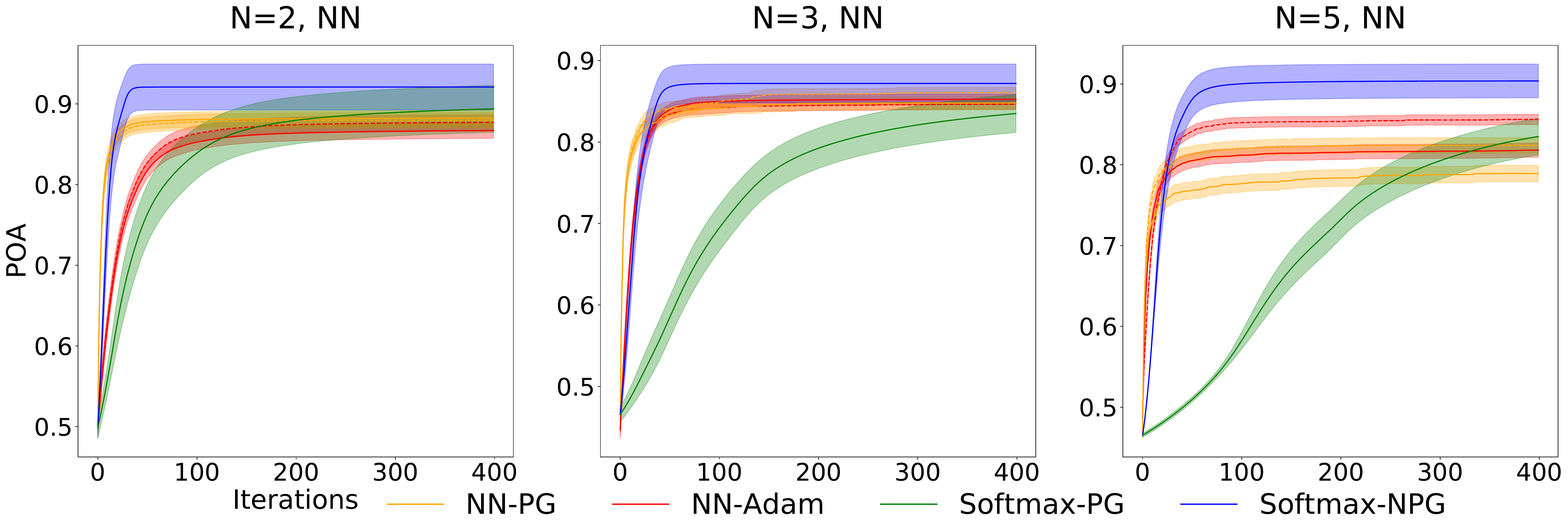}
    \end{subfigure}
    \begin{subfigure}{\textwidth}
    \centering
    \includegraphics[width=.91\textwidth]{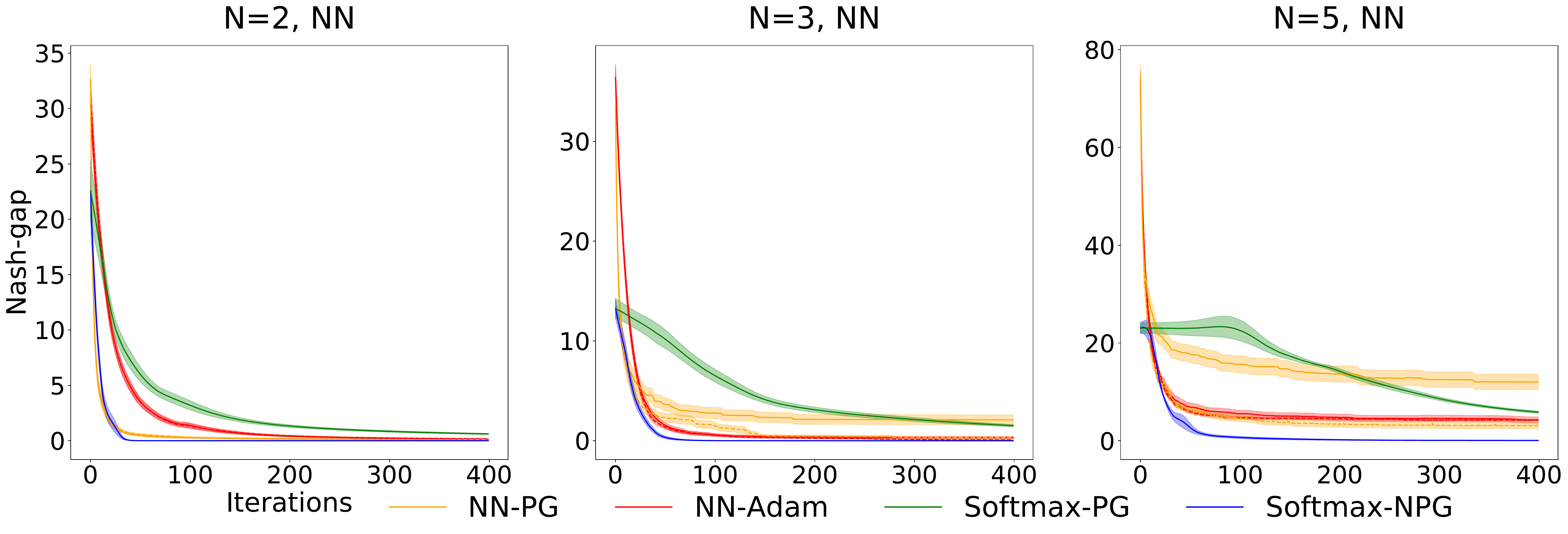}
    \end{subfigure}
    \caption{POA (top) and Nash-gap (bottom) under the NN parameterization (means and standard errors over 20 random initializations).The dashed lines are the curves of the log barrier regularized version of the algorithms with the same color. }
    \label{fig:NN-POA-Nash-gap}
\end{figure*}

\section{Experiments}
\label{sec:Experiments}
\textbf{Environment.} We evaluate the algorithms on Coordination Game, which extends the two players version in \cite{zhang2021gradient} to multiple players $N=2,3,5$. The state space and action space are $\mathcal{S} = \mathcal{S}^1\times\cdots\times\mathcal{S}^N, \mathcal{A} = \mathcal{A}^1\times\cdots\times\mathcal{A}^N$, respectively, where $\forall i\leq N, \mathcal{S}^i \in\{0,1\},\mathcal{A}^i \in\{0,1\}$. The reward is shared by all the agents (cooperative setting, a special case of Markov Potential Games), and it encourages agents to be in the same local state. To have rewards with more different levels, we design the reward in the way that when the number of  agents occupy local state $0$ or $1$, whichever the maximum, to be the same, the state with more local states of $0$s is larger than the one with more $1$s.  The transition function for each agent $i$'s local state is $P(s^i=0|a^i=0)=1-\epsilon$, $P(s^i=0|a^i=1)=\epsilon$, where $\epsilon=0.1$.

\textbf{Algorithms.} We exhaustively evaluate the performance of policy gradient \textbf{PG}, natural policy gradient \textbf{NPG}, and best response natural policy gradient \textbf{NPG-BR} with softmax parameterization under the tabular setting, w/wo log barrier regularizer. Besides softmax parameterization for PG, we also consider the neural network parameterization with softmax activation in the last layer \textbf{NN-PG}. Precisely, the policy gradient update rule for agent $i$'s neural network policy $\pi^i_{\theta^i}\colon \mathcal{S}\to \Delta(\mathcal{A}^i)$ is 
\begin{align*}
\nabla_{\theta^i} \Phi_{\theta}(\mu)=
\sum_s\sum_{a^i}d_\mu^{\pi_\theta}(s)\pi^i_{\theta^i}(a^i|s)A^i_\theta(s,a^i)\nabla_{\theta^i}\pi_{\theta^i}^i(a^i|s)
.
\end{align*}

We run each algorithm in Coordination Game with $N=2,3,5$ agents and plot the Nash-gap and POA as the evaluation metrics.
The algorithms, both the tabular softmax and the neural network parameterizations, share the same initial policy parameters, which are sampled from the normal distribution of mean 0 and standard deviation 1.
For each log barrier regularized algorithm, we performed a grid search for its coefficient $\lambda\in\{0.01, 0.1, 1.0, 10.0, 100.0\}$ and picked the one with the best POA.
Additional details of our experiment are presented in Appendix \ref{sec:Experiment details}.

\subsection{Results under the tabular softmax parameterization}

Figure \ref{fig:Softmax-POA-Nash-gap} presents the POA and the Nash-gap of the algorithms under the tabular softmax parameterization.
The results help address the following questions:

\textit{How fast do the algorithms converge?}
In terms of both the POA and the Nash-gap, NPG converges fastest, with NPG-BR the second and PG the slowest.
This result demonstrates the improvement in the convergence rate of using the natural policy gradient over the policy gradient.

\textit{What is the effect of $K$ for NPG-BR?}
We did a grid search of $K\in\{1,5,10,20,50\}$ for NPG-BR (details in Appendix \ref{sec:Effect of $K$ for the NPG-BR dynamics}), we show the results for the best-performing $K$ in terms of the POA for $N=2,3,5$ separately in Figure \ref{fig:Softmax-POA-Nash-gap}.
We observe that $K=5$ is the best for $N=3,5$ and $K=50$, the largest value we searched, is the best for $N=2$.

\textit{How do the algorithms compare in terms of the POA?}
Consistent with the converge rate, NPG enjoys the overall highest POA, with NPG-BR the second and PG the lowest.

\subsection{Results under the neural network parameterization}
\label{sec:Results under the neural network parameterization}
Figure \ref{fig:NN-POA-Nash-gap} presents the POA and the Nash-gap of the algorithms under the neural network (NN) parameterization.
The results help address the following questions:

\textit{Does NN help improve the convergence/POA from tabular softmax?}
With the NN parameterization, the PG algorithm (``NN-PG'') significantly outperforms its tabular softmax counterpart (``Softmax-PG'') in terms of both the convergence rate and the POA.
NN-PG even outperforms Softmax-NPG in terms of POA at the beginning of the training, although and eventually the POA of Softmax-NPG is the highest among all.
This demonstrates the significant improvement of the NN parameterization over the tabular softmax.

\textit{What is the effect of the NN regularization?}
Compared with the results under tabular softmax, the log barrier regularization under NN has a significantly larger impact: it both improves the POA and reduces the Nash-gap at convergence, especially when $N$ is large (e.g., $N=5$).

\textit{What is the effect of the NN optimizer?}
Among all NN variants, NN-PG is the best in terms of POA when $N$ is small, and the regularized NN-Adam is the best when $N$ is large.
When $N=5$, the POA of the best NN variant, the regularized NN-Adam, is still significantly smaller than Softmax-NPG.

\section{Conclusion and discussion}
\label{sec:Discussion}
To conclude, we have established in Section \ref{sec:Convergence of the tabular softmax policy gradient in MPGs} convergence to (near-)Nash policies in Markov potential games of several policy gradient-based dynamics under tabular softmax parameterization, including 
asymptotic convergence of the standard policy gradient dynamics (Section \ref{sec:Asymptotic convergence of the policy gradient dynamics}),
its finite-time convergence with log-barrier regularization (Section \ref{sec:Policy gradient dynamics with log-barrier regularization}),
and finite-time convergence of the approximate best-response natural policy gradient dynamics (Section \ref{sec:Approximate best-response natural policy gradient dynamics}).
In Section \ref{sec:Bounding the price of anarchy in smooth Markov (potential) games}, we have extended the notion of smoothness in normal-form games to Markov games and established the price-of-anarchy bounds of near-Nash policies in smooth Markov games and of the approximate maximum-gain best-response dynamics in smooth Markov potential games.

\paragraph{Future work.}
(i) Our theoretical guarantee for the NPG dynamics is limited to the (approximate) best-response variant, although our empirical results imply that the standard NPG dynamics where all agents get updated per iteration should also converge. This suggests that a future direction is to establish the convergence of the standard NPG dynamics.
(ii) Our POA bound is also limited to the (approximate) best-response NPG dynamics, and a future direction is to provide POA bounds for other learning dynamics.
(iii) Both the theoretical and the empirical parts of this paper are limited to exact gradient computation, and therefore an immediate future direction is to explore sample-based learning dynamics.

\bibliography{ref}
\bibliographystyle{icml2022}

\newpage
\appendix
\onecolumn
\numberwithin{theorem}{section} 

\section{Proof of Lemma \ref{lemma:state-based tabular softmax multi-agent policy gradient}}
\label{sec:Proof of Lemma lemma:state-based tabular softmax multi-agent policy gradient}
Note that
\begin{align*}
    \frac{\partial\log{\pi^i_{\theta^i}(a^i_\prime|s')}}{\partial \theta^i_{s, a^i}} = 
    \mathbbm{1}[s = s'](\mathbbm{1}[a^i=a^i_\prime]-\pi^i_{\theta^i}(a^i|s)).
\end{align*}
Plugging it and by similar derivations in the proof of Lemma C.1 in \cite{agarwal2019theory}, we have:
\begin{align*}
    \frac{\partial V^i_\theta(\mu)}{\partial \theta^i_{s,a^i} } =& \E_{s'\sim d^{\pi_\theta}_\mu} \E_{a'\sim\pi_\theta(\cdot|s)}\left[\mathbbm{1}[(s',a^i_\prime)=(s,a^i)]A^i_\theta(s',a')\right] \\
    =& d^{\pi_\theta}_\mu(s) \pi^i_{\theta^i}(a^i|s)\E_{a^{-i}\sim\pi^{-i}_{\theta^{-i}}(\cdot|s)}\left[A^i_\theta(s, a^i, a^{-i})\right].
\end{align*}
This concludes the proof.

\section{Proof of Lemma \ref{lemma:smoothness_tabular_softmax}}
\label{sec:Proof of Lemma lemma:smoothness_tabular_softmax}
Since $\Phi_\theta(s_0)$, abbreviated as $\Phi_\theta$ in this proof, is (assumed to be) twice-differentialble, as an equivalent condition for smoothness, we will bound the spectral norm of its Hessian $\nabla^2_\theta \Phi_\theta$.
Similar to the proof of Lemma 4.4 in \cite{leonardos2021global}, we view Hessian
\begin{align*}
  \nabla^2_\theta \Phi_\theta=\left[
  \frac{\partial^2 \Phi_\theta}{\partial\theta^i_{s,a^i}\partial\theta^j_{s',a^j}}\right]_{i,s,a^i,j,s',a^j}
\end{align*}
as a symmetric $N \times N$ block matrix with submatrices
\begin{align*}
    \nabla^2_{\theta^i\theta^j} \Phi_\theta=\left[
  \frac{\partial^2 \Phi_\theta}{\partial\theta^i_{s,a^i}\partial\theta^j_{s',a^j}}\right]_{s,a^i,s',a^j}
\end{align*}
for all $i,j\in\mathcal{N}$. Claim C.2 in \cite{leonardos2021global} shows that if we can bound the spectral norm of any submatrix as $\norm{\nabla^2_{\theta^i\theta^j} \Phi_\theta}_2\leq L$, then the spectral norm of the block matrix is bounded as $\norm{\nabla^2_\theta \Phi_\theta}_2\leq NL$. We then next bound the spectural norm (i.e., the largest absolute eigenvalue) of matrix $\nabla^2_{\theta^j\theta^i} \Phi_\theta$.
Noting $\nabla^2_{\theta^j\theta^i} \Phi_\theta=\nabla^2_{\theta^j\theta^i} V^j_\theta=\nabla^2_{\theta^j\theta^i} V^i_\theta$ due to \eqref{eq:total potential function}, it suffices to define $U(t):=V_{\theta^i + t\cdot u,\theta^{-i}}$ and $W(t,s):=V_{\theta^i + t\cdot u,\theta^j + s\cdot v,\theta^{-i,-j}}$ for scalars $t,s\geq0$ and unit vectors $u,v$, and to show
\begin{align*}
    \max_{\norm{u}_2=1} \abs{\left.\frac{d^2U(t)}{dt^2}\right|_{t=0}} \leq \frac{41}{4(1-\gamma)^3}
    \text{ and }
    \max_{\norm{u}_2=\norm{v}_2=1} \abs{\left.\frac{d^2W(t,s)}{dtds}\right|_{t=0, s=0}} \leq \frac{41}{4(1-\gamma)^3}
    .
\end{align*}

For $U(t)$, we decompose it as $U(t) = \sum_{a^i}\sum_{a^{-i}}\pi^i_{\theta^i + t\cdot u}(a^i|s_0)\cdot\pi^{-i}_{\theta^{-i}}(a^{-i}|s_0)\cdot Q_{\theta^i + t\cdot u,\theta^{-i}}(s_0,a^i,a^{-i})$.
Abbreviating $\pi^i_{\theta^i + t\cdot u}$ as $\pi^i_t$, $\pi^{-i}_{\theta^{-i}}$ as $\pi^{-i}$, and $Q_{\theta^i + t\cdot u,\theta^{-i}}$ as $Q_t$, we have
\begin{align*}
    \frac{d^2U(t)}{dt^2} = \sum_{a^i}\sum_{a^{-i}} \bigg(
    &\frac{d^2\pi^i_t(a^i|s_0)}{dt^2} \cdot \pi^{-i}(a^{-i}|s_0) \cdot Q_t(s_0,a^i,a^{-i}) \\
    +& 2\frac{d\pi^i_t(a^i|s_0)}{dt} \cdot \pi^{-i}(a^{-i}|s_0) \cdot \frac{dQ_t(s_0,a^i,a^{-i})}{dt} \\
    +& \pi^i_t(a^i|s_0) \cdot \pi^{-i}(a^{-i}|s_0) \cdot \frac{d^2Q_t(s_0,a^i,a^{-i})}{dt^2}\bigg)
\end{align*}
We then bound $\abs{\left.\frac{d^2U(t)}{dt^2}\right|_{t=0}}$ for any unit vector $u$ by bounding the three terms, respectively.
For the first term, we have $\sum_{a^i}\abs{\left.\frac{d^2\pi^i_t(a^i|s_0)}{dt^2}\right|_{t=0}}\leq6=:C_2$ as proved in Lemma D.4 in \cite{agarwal2019theory}, $0 \leq Q_t(s_0,a^i,a^{-i})\leq\frac{1}{1-\gamma}$ assuming the reward is bounded in $[0,1]$, and $\sum_{a^{-i}}\pi^{-i}(a^{-i}|s_0)=1$.
For the second term, we have $\sum_{a^i}\abs{\left.\frac{d\pi^i_t(a^i|s_0)}{dt}\right|_{t=0}}\leq2=:C_1$ as proved in Lemma D.4 in \cite{agarwal2019theory}, and $\abs{\left.\frac{dQ_t(s_0,a^i,a^{-i})}{dt}\right|_{t=0}}\leq\frac{\gamma C_1}{(1-\gamma)^2}$ as proved in Lemma D.2 in \cite{agarwal2019theory} and Lemma 4.4 in \cite{leonardos2021global}.
For the third term, we have $\abs{\left.\frac{d^2Q_t(s_0,a^i,a^{-i})}{dt^2}\right|_{t=0}}\leq\frac{2\gamma^2C_1}{(1-\gamma)^3} + \frac{\gamma C_2}{(1-\gamma)^2}$ as proved in Lemma D.2 in \cite{agarwal2019theory}. We hence derive the bound:
\begin{align*}
    \max_{\norm{u}_2=1} \abs{\left.\frac{d^2U(t)}{dt^2}\right|_{t=0}} \leq&
    \frac{C_2}{1-\gamma} + \frac{2\gamma C_1^2}{(1-\gamma)^2} + \frac{2\gamma^2C_1}{(1-\gamma)^3} + \frac{\gamma C_2}{(1-\gamma)^2} \\
    =& \frac{C_2}{(1-\gamma)^2} + \frac{2\gamma C_1^2}{(1-\gamma)^3}
    = \frac{6+2\gamma}{(1-\gamma)^3} \qquad \text{($C_1=2, C_2=6$)}\\
    \leq&  \frac{8}{(1-\gamma)^3}
    \leq  \frac{41}{4(1-\gamma)^3}
\end{align*}

For $W(t,s)$, similarly, we decompose it as $W(t,s) = \sum_{a^i}\sum_{a^j}\sum_{a^{-i,-j}}\pi^i_{\theta^i + t\cdot u}(a^i|s_0)\cdot\pi^j_{\theta^j + s\cdot v}(a^j|s_0)\cdot\pi^{-i,-j}_{\theta^{-i,-j}}(a^{-i,-j}|s_0)\cdot Q_{\theta^i + t\cdot u,\theta^j + s\cdot v,\theta^{-i,-j}}(s_0,a^i,a^j,a^{-i,-j})$. 
With similar abbreviations, we have 
\begin{align*}
    \frac{d^2W(t,s)}{dtds}= \sum_{a^i}\sum_{a^j}\sum_{a^{-i,-j}} \bigg(
     &\frac{d\pi^i_t(a^i|s_0)}{dt} \cdot \frac{d\pi^j_s(a^j|s_0)}{ds} \cdot \pi^{-i,-j}(a^{-i,-j}|s_0) \cdot Q_{t,s}(s_0,a^i,a^j,a^{-i,-j}) \\
    +& \frac{d\pi^i_t(a^i|s_0)}{dt} \cdot \pi^j_s(a^j|s_0) \cdot \pi^{-i,-j}(a^{-i,-j}|s_0) \cdot \frac{dQ_{t,s}(s_0,a^i,a^j,a^{-i,-j})}{ds} \\
    +& \pi^i_t(a^i|s_0) \cdot \frac{d\pi^j_s(a^j|s_0)}{ds} \cdot \pi^{-i,-j}(a^{-i,-j}|s_0) \cdot \frac{dQ_{t,s}(s_0,a^i,a^j,a^{-i,-j})}{dt}\\
    +& \pi^i_t(a^i|s_0) \cdot \pi^j_s(a^j|s_0) \cdot \pi^{-i,-j}(a^{-i,-j}|s_0) \cdot \frac{d^2Q_{t,s}(s_0,a^i,a^j,a^{-i,-j})}{dtds}
    \bigg)
    .
\end{align*}
We then bound $\abs{\left.\frac{d^2W(t,s)}{dtds}\right|_{t=0, s=0}}$ for any unit vectors $u,v$ by bounding the four terms, respectively.
Similarly, the first term can be bounded by $\frac{C_1^2}{1-\gamma}$, the second term by $\frac{\gamma C_1^2}{(1-\gamma)^2}$, the third term by $\frac{\gamma C_1^2}{(1-\gamma)^2}$, and the fourth term by $\frac{C_2}{(1-\gamma)^2} + \frac{2\gamma C_1^2}{(1-\gamma)^3}$.
We hence derive the bound:
\begin{align*}
   \max_{\norm{u}=\norm{v}=1}\abs{\left.\frac{d^2W(t,s)}{dtds}\right|_{t=0, s=0}} \leq&
   \frac{C_1^2}{1-\gamma}
   + \frac{\gamma C_1^2}{(1-\gamma)^2}
   + \frac{\gamma C_1^2}{(1-\gamma)^2}
   + \frac{C_2}{(1-\gamma)^2} + \frac{2\gamma C_1^2}{(1-\gamma)^3}\\
   =& \frac{-4\gamma^2+2\gamma+10}{(1-\gamma)^3} \qquad\text{($C_1=2, C_2=6$)}\\
   \leq& \frac{41}{4(1-\gamma)^3}
   .
\end{align*}
This concludes the proof.

\section{Proof of Theorem \ref{theorem:Asymptotic convergence to Nash with gradient ascent}}
\label{sec:Proof of Theorem Asymptotic convergence to Nash with gradient ascent}
\subsection{Notation}
Define $$V_\phi^{\pi_\theta}(s)=\E[\sum_{t=0}^\infty \gamma^t\phi(s_t,a_t)|\pi_\theta,s_0=s]$$
$$\Phi^{\pi_\theta}(\mu)=\E_{s_0\sim\mu}[V_\phi^{\pi_\theta}(s_0)]$$
$$Q_\phi^{\pi_\theta}(s,a)=\E[\sum_{t=0}^\infty \gamma^t\phi(s_t,a_t)|\pi_\theta,s_0=s,a_0=a]$$
$$A_\phi^{\pi_\theta}(s,a)=Q_\phi^{\pi_\theta}(s,a)-V_\phi^{\pi_\theta}(s)$$
Suppose $\Phi_{\text{min}}\leq Q_\phi^{\pi_\theta}(s,a)\leq\Phi_{\text{max}}$.

\subsection{Smoothness of F}
\begin{lemma}[Smoothness of $F$ under tabular softmax]
\label{lemma:smoothness_of_F_tabular_softmax}
Fix a state $s$. Let $\theta_s = [(\theta^1_s)^\top,...,(\theta^N_s)^\top]^\top\in\mathbb{R}^{\sum_i|\mathcal{A}^i|}$ be the column vector of parameters for state $s$, with $\theta^i_s\in\mathbb{R}^{|\mathcal{A}^i|}$ for $i\in\mathcal{N}$. For some fixed vector $c_s\in\mathbb{R}^{|\mathcal{A}|}$, define $F_s(\theta_s) := \sum_{a\in\mathcal{A}} \pi_{\theta_s}(a|s)c_{s,a}=:\pi_{\theta_s} \cdot c_s$ with $\pi_{\theta_s}\in\mathbb{R}^{|\mathcal{A}|}$ and $\cdot$ denoting inner product.
Then, $F_s(\theta_s)$ is -smooth.
\end{lemma}
\begin{proof}
We will view Hessian $\nabla^2_{\theta_s}F_s(\theta_s)$ as a $N \times N$ block matrix and bound the spectral norm of each submatrix as  $\norm{\nabla^2_{\theta^i_s \theta^j_s}F_s(\theta_s)}_2\leq L$, which bounds the Hessian's spectral norm as $\norm{\nabla^2_{\theta_s}F_s(\theta_s)}_2\leq NL$.

We have
\begin{align*}
   \nabla_{\theta^i_s}F_s(\theta_s) =  \nabla_{\theta^i_s}(\pi_{\theta_s} \cdot c_s) =  (\nabla_{\theta^i_s}\pi_{\theta_s})^\top c_s =  \nabla_{\theta^i_s}\pi^i_{\theta^i_s}\left(\pi^{-i}_{\theta^{-i}_s}\otimes I_{|\mathcal{A}^i|}\right)^\top M^i c_s
\end{align*}
where $\nabla_{\theta^i_s}F_s(\theta_s)\in\mathbb{R}^{1\times|\mathcal{A}^i|}$, $M^i\in\mathbb{R}^{|\mathcal{A}|\times|\mathcal{A}|}$ is the permutation matrix that permutes all joint actions to be sorted as $a=(a^{-i}, a^i)$, $I_n$ is the $n \times n$ identity matrix, and $\otimes$ is the Kronecker product.
For the tabular softmax parameterization, we have
\begin{align*}
    \nabla_{\theta^i_s}\pi^i_{\theta^i_s} = {\rm diag}\left(\pi^i_{\theta^i_s}\right) - \pi^i_{\theta^i_s}\left(\pi^i_{\theta^i_s}\right)^\top
    .
\end{align*}
The submatrix is therefore
\begin{align*}
    \nabla^2_{\theta^i_s\theta^j_s}F_s(\theta_s) =& \nabla_{\theta^j_s} \left(\nabla_{\theta^i_s}\pi^i_{\theta^i_s}\left(\pi^{-i}_{\theta^{-i}_s}\otimes I_{|\mathcal{A}^i|}\right)^\top M^i c_s\right)\\
\end{align*}
If $j = i$:
\begin{align*}
    \nabla^2_{\theta^i_s\theta^i_s}F_s(\theta_s) =& \nabla_{\theta^i_s} \left(\nabla_{\theta^i_s}\pi^i_{\theta^i_s} \left(\pi^{-i}_{\theta^{-i}_s}\otimes I_{|\mathcal{A}^i|}\right)^\top M^i c_s\right)\\
    =& \nabla_{\theta^i_s}(\pi^i_{\theta^i_s} \odot b - (\pi^i_{\theta^i_s}\cdot b)\pi^i_{\theta^i_s})
\end{align*}
, where $b=\left(\pi^{-i}_{\theta^{-i}_s}\otimes I_{|\mathcal{A}^i|}\right)^\top M^i c_s.$
\\For the first term, we get 
$$\nabla_{\theta^i_s}(\pi^i_{\theta^i_s} \odot b)=\text{diag}(\pi^i_{\theta^i_s} \odot b)-\pi^i_{\theta^i_s}(\pi^i_{\theta^i_s} \odot b)^\top$$
For the second term we get:
$$\nabla_{\theta^i_s}((\pi^i_{\theta^i_s}\cdot b)\pi^i_{\theta^i_s})=(\pi^i_{\theta^i_s}\cdot b)\nabla_{\theta^i_s}(\pi^i_{\theta^i_s})+(\nabla_{\theta^i_s}(\pi^i_{\theta^i_s}\cdot b))(\pi^i_{\theta^i_s})^\top$$
$$\longrightarrow \nabla^2_{\theta^i_s\theta^i_s}F_s(\theta_s) = \text{diag}(\pi^i_{\theta^i_s} \odot b)-\pi^i_{\theta^i_s}(\pi^i_{\theta^i_s} \odot b)^\top-(\pi^i_{\theta^i_s}\cdot b)\nabla_{\theta^i_s}(\pi^i_{\theta^i_s})-(\nabla_{\theta^i_s}(\pi^i_{\theta^i_s}\cdot b))(\pi^i_{\theta^i_s})^\top$$
Since 
\begin{align*}
    &\text{max}(\norm{\text{diag}(\pi^i_{\theta^i_s} \odot b)}_2,\norm{\pi^i_{\theta^i_s}\odot b}_2,|\pi^i_{\theta^i_s}\cdot b|) \leq \norm{b}_\infty= \norm{c}_\infty &\\
    &\norm{\nabla_{\theta^i_s}\pi^i_{\theta^i_s}}_2 = \norm{{\rm diag}\left(\pi^i_{\theta^i_s}\right) - \pi^i_{\theta^i_s}\left(\pi^i_{\theta^i_s}\right)^\top}_2\leq 1&\\
    & \norm{\nabla_{\theta^i_s}(\pi^i_{\theta^i_s}\cdot b)}_2\leq \norm{\pi^i_{\theta^i_s} \odot b}_2+\norm{(\pi^i_{\theta^i_s}\cdot b)\pi^i_{\theta^i_s}}_2\leq 2\norm{c}_\infty,
\end{align*}
we know that $$\norm{\nabla^2_{\theta^i_s\theta^i_s}F_s(\theta_s)}_2\leq 5\norm{c}_\infty$$

If $j \neq i$: 
\begin{align*}
    \nabla^2_{\theta^i_s\theta^j_s}F_s(\theta_s) =& 
    M^j\nabla_{\theta^i_s}\pi^i_{\theta^i_s}\Bigg(\bigg(\Big(\pi^{-i,-j}_{\theta^{-i,-j}_s}\otimes I_{|\mathcal{A}^j|}\Big)\nabla_{\theta^j_s}\pi^j_{\theta^j_s}\bigg)\otimes I_{|\mathcal{A}^i|}\Bigg)^\top M^jM^ic_s
\end{align*}
Since 
\\$\norm{M^jM^ic_s}_2 \leq \sqrt{N}\norm{c}_\infty$
\begin{equation} \label{eq1}
\begin{split}
\norm{\bigg(\Big(\pi^{-i,-j}_{\theta^{-i,-j}_s}\otimes I_{|\mathcal{A}^j|}\Big)\nabla_{\theta^j_s}\pi^j_{\theta^j_s}\bigg)\otimes I_{|\mathcal{A}^i|}}_2 & = \norm{\Big(\pi^{-i,-j}_{\theta^{-i,-j}_s}\otimes I_{|\mathcal{A}^j|}\Big)\nabla_{\theta^j_s}\pi^j_{\theta^j_s}}_2 \\
 & \leq \norm{\pi^{-i,-j}_{\theta^{-i,-j}_s}\otimes I_{|\mathcal{A}^j|}}_2\norm{\nabla_{\theta^j_s}\pi^j_{\theta^j_s}}_2
 \\
 & \leq\norm{\pi^{-i,-j}_{\theta^{-i,-j}_s}}_2
  \\
 & \leq 1
\end{split}
\end{equation}
we know that 
$$\nabla^2_{\theta^i_s\theta^j_s}F_s(\theta_s)\leq \sqrt{N}\norm{c}_\infty$$
\end{proof}
Therefore, we have 
$$\norm{\nabla^2_{\theta_s}F_s(\theta_s)}_2\leq N\text{max}(5,\sqrt{N})\norm{c}_\infty$$

\begin{lemma}
For product policy that can be factorized into the product of individual policies with softmax parameterization, we have: 
$$\frac{\partial V^{\pi_\theta}(\mu)}{\partial \theta_{s,a^i}^{i}}=\frac{\partial \Phi^{\pi_\theta}(\mu)}{\partial \theta_{s,a^i}^{i}}=\frac{1}{1-\gamma}d_\mu^{\pi_\theta}(s)\pi_{\theta^{i}}(a^i|s)A_\phi^{\pi_{\theta},i}(s,a^i)$$
,where $Q_\phi^{\pi_\theta,i}(s,a^i)=\E_{a^{-i}\sim\pi_{\theta^{-i}}(\cdot|s)}\Big[Q_\phi^{\pi_\theta}(s,a^i,a^{-i})\Big],A_\phi^{\pi_\theta,i}(s,a^i)=Q_\phi^{\pi_\theta,i}(s,a^i)-V_\phi^{\pi_\theta}(s).$
\end{lemma}

\begin{proof}
$$\frac{\partial V_\phi^{\pi_\theta}(\mu)}{\partial \theta_{s^{\prime},a^i}^{i}}=\frac{\partial \Phi^{\pi_\theta}(\mu)}{\partial \theta_{s,a^i}^{i}}=\frac{1}{1-\gamma}\E_{s\sim d_\mu^{\pi_\theta}}\E_{a\sim\pi_{\theta}(\cdot|s)}\Big[A_\phi^{\pi_\theta}(s,a)\frac{\partial \log \pi_{\theta^{i}}(a^i|s)}{\partial \theta_{s^\prime,a^i}^i}\Big]$$
$$=\frac{1}{1-\gamma}\E_{s\sim d_\mu^{\pi_\theta}}\E_{a\sim\pi_{\theta}(\cdot|s)}\Big[A_\phi^{\pi_\theta}(s,a)\mathbbm{1}[s=s^\prime](\mathbbm{1}[a[m]=a^i]-\pi_{\theta^{i}}(a^i|s))\Big]$$
$$=\frac{1}{1-\gamma}d_\mu^{\pi_\theta}(s^\prime)\E_{a\sim\pi_{\theta}(\cdot|s^\prime)}\Big[A_\phi^{\pi_\theta}(s^\prime,a)(\mathbbm{1}[a[m]=a^i]-\pi_{\theta^{i}}(a^i|s^\prime))\Big]$$
$$=\frac{1}{1-\gamma}d_\mu^{\pi_\theta}(s^\prime)(\E_{a\sim\pi_{\theta}(\cdot|s^\prime)}\Big[A_\phi^{\pi_\theta}(s^\prime,a)\mathbbm{1}[a[m]=a^i]\Big]-\E_{a\sim\pi_{\theta}(\cdot|s^\prime)}\Big[A_\phi^{\pi_\theta}(s^\prime,a)\pi_{\theta^{i}}(a^i|s^\prime)\Big])$$
$$=\frac{1}{1-\gamma}d_\mu^{\pi_\theta}(s^\prime)(\E_{a\sim\pi_{\theta}(\cdot|s^\prime)}\Big[A_\phi^{\pi_\theta}(s^\prime,a)\mathbbm{1}[a[m]=a^i]\Big]-\pi_{\theta^{i}}(a^i|s^\prime)\E_{a\sim\pi_{\theta}(\cdot|s^\prime)}\Big[A_\phi^{\pi_\theta}(s^\prime,a)\Big])$$
$$=\frac{1}{1-\gamma}d_\mu^{\pi_\theta}(s^\prime)\E_{a\sim\pi_{\theta}(\cdot|s^\prime)}\Big[A_\phi^{\pi_\theta}(s^\prime,a)\mathbbm{1}[a[m]=a^i]\Big]$$
$$=\frac{1}{1-\gamma}d_\mu^{\pi_\theta}(s^\prime)\sum_{a}\pi_{\theta}(a|s^\prime)A_\phi^{\pi_\theta}(s^\prime,a)\mathbbm{1}[a[m]=a^i]$$
$$=\frac{1}{1-\gamma}d_\mu^{\pi_\theta}(s^\prime)\pi_{\theta^{i}}(a^i|s^\prime)\E_{a^{-i}\sim\pi_{\theta}(\cdot|s^\prime)}\Big[A_\phi^{\pi_\theta}(s^\prime,a^i,a^{-i})\Big]$$
$$=\frac{1}{1-\gamma}d_\mu^{\pi_\theta}(s^\prime)\pi_{\theta^{i}}(a^i|s^\prime)A_\phi^{\pi_{\theta},i}(s^\prime,a^i)$$
\end{proof}

\begin{lemma}
For all agents $i$ with a round of parallel update  $$\theta^{i}_{t+1}=\theta^{i}_{t}+\eta\nabla 
V^i_{\theta^i_t}(\mu)=\theta^{i}_{t}+\eta\nabla 
\Phi_{\theta^i_t}(\mu)$$ with learning rates $\eta\leq \frac{1-\gamma}{\beta}$, where 
$\beta=NL(\Phi_{\text{max}}-\Phi_{\text{min}})$ , $L=\text{max}(5,\sqrt{N})$, we have $$V_\phi^{(t+1)}(s)\geq V_\phi^{(t)}(s);Q_\phi^{(t+1)}(s,a)\geq Q_\phi^{(t)}(s,a).$$
\end{lemma}

\begin{proof}
Let us use the notation $\theta_s\in\mathbb{R}^{\sum_i|\mathcal{A}_\phi^i|}$ to refer to the parameters of the product policy on state $s$. Define $$F_s(\theta_s)=\sum_{a}\pi_{\theta_s}(a|s)c(s,a)$$
where $c(s,a)$ is treated as a constant, and is set to be $A_\phi^{(t)}(s,a)$ later in the proof. Thus,
$$\frac{\partial F_s(\theta_s)}{\partial \theta_{s,a^i}^{i}}\bigg|_{\theta_{s}^{t,i}}=\sum_{a^\prime}\frac{\partial \pi_{\theta_s}(a^\prime|s)}{\partial \theta_{s,a^i}^{i}}\bigg|_{\theta_{s}^{t,i}}c(s,a^\prime)$$
$$=\underbrace{\sum_{a^\prime}\mathbbm{1}[a^\prime[i]=a^i]\frac{\partial \pi_{\theta_s}(a^\prime|s)}{\partial \theta_{s,a^i}^{i}}\bigg|_{\theta_{s}^{t,i}}c(s,a^\prime)}_{(1)}+\underbrace{\sum_{a^\prime}\mathbbm{1}[a^\prime[i]\neq a^i]\frac{\partial \pi_{\theta_s}(a^\prime|s)}{\partial \theta_{s,a^i}^{i}}\bigg|_{\theta_{s}^{t,i}}c(s,a^\prime)}_{(2)}$$
$$(1)=\sum_{a^\prime}\mathbbm{1}[a^\prime[i]=a^i]\frac{\pi_{\theta_s}(a^\prime|s)}{\pi_{\theta_s^{i}}(a^i|s)}\Big[\pi_{\theta_s^{i}}(a^i|s)(1-\pi_{\theta_s^{i}}(a^i|s))\Big]c(s,a^\prime)$$
$$=\sum_{a^\prime}\mathbbm{1}[a^\prime[i]=a^i]\pi_{\theta_s}(a^\prime|s)\Big(1-\pi_{\theta_s^{i}}(a^i|s)\Big)\bigg|_{\theta_{s,t}}c(s,a^\prime)$$
$$=\sum_{a^\prime}\mathbbm{1}[a^\prime[i]=a^i]\pi^{t}(a^\prime|s)\Big(1-\pi^{t,i}(a^i|s)\Big)c(s,a^\prime)$$
$$(2)=\sum_{a^\prime}\mathbbm{1}[a^\prime[i]\neq a^i]\frac{\pi_{\theta_s}(a^\prime|s)}{\pi_{\theta_s^{i}}(a^i|s)}\Big(-\pi_{\theta_s^{i}}(a^i|s)\pi_{\theta_s^{i}}(a^\prime[i]|s)\Big)\bigg|_{\theta_{s}^{t,i}}c(s,a^\prime)$$
$$=-\sum_{a^\prime}\mathbbm{1}[a^\prime[i]\neq a^i]\pi_{\theta_s}(a^\prime|s)\pi_{\theta_s^{i}}(a^i|s)\bigg|_{\theta_{s}^{t,i}}c(s,a^\prime)$$
$$=-\sum_{a^\prime}\mathbbm{1}[a^\prime[i]\neq a^i]\pi^{t}(a^\prime|s)\pi^{t,i}(a^i|s)c(s,a^\prime)$$
$(1)+(2)=\sum_{a^\prime}\mathbbm{1}[a^\prime[i]=a^i]\pi^{t}(a^\prime|s) c(s,a^\prime)-$

$$\Big(\sum_{a^\prime}\mathbbm{1}[a^\prime[i]=a^i]\pi^{t}(a^\prime|s)\pi^{t,i}(a^i|s)c(s,a^\prime)+\sum_{a^\prime}\mathbbm{1}[a^\prime[i]\neq a^i]\pi^{t}(a^\prime|s)\pi^{t,i}(a^i|s)c(s,a^\prime)\Big)$$
$$=\sum_{a^\prime}\mathbbm{1}[a^\prime[i]=a^i]\pi^{t}(a^\prime|s) c(s,a^\prime)-\sum_{a^\prime}\pi^{t}(a^\prime|s)\pi^{t,i}(a^i|s)c(s,a^\prime)$$
Let $c(s,a^\prime)=A_\phi(s,a^\prime)$,
$$=\sum_{a^\prime}\mathbbm{1}[a^\prime[i]=a^i]\pi^{t}(a^\prime|s) A_\phi(s,a^\prime)-\sum_{a^\prime}\pi^{t}(a^\prime|s)\pi^{t,i}(a^i|s)A_\phi(s,a^\prime)$$
$$=\sum_{a^\prime}\mathbbm{1}[a^\prime[i]=a^i]\pi^{t}(a^\prime|s) A_\phi(s,a^\prime)$$
$$=\pi_{\theta^{i}}(a^i|s)A_\phi^{\pi_{\theta},i}(s,a^i)$$
Therefore, $$\nabla \Phi_{\theta^i}^{t}(\mu)=\frac{1}{1-\gamma}d_\mu^{\pi_\theta}(s)\frac{\partial F_s(\theta_s)}{\partial \theta_{s,a^i}^{i}}$$
$$\longrightarrow \theta_s^{t+1}=\theta_s^{t}+\eta \frac{1}{1-\gamma}d_\mu^{\pi_\theta}(s)\frac{\partial F_s(\theta_s)}{\partial \theta_{s}}\bigg|_{\theta_s^{t}}$$
Since $F_s(\theta_s)$ is a $\beta$-smooth function for $\beta=N\text{max}(5,\sqrt{N})(\Phi_{\text{max}}-\Phi_{\text{min}})$, then our assumptions that $\eta\leq \frac{1-\gamma}{\beta}=\frac{1-\gamma}{N\text{max}(5,\sqrt{N})(\Phi_{\text{max}}-\Phi_{\text{min}})}$ implies $\eta \frac{1}{1-\gamma}d_\mu^{\pi_\theta}(s)\leq \frac{1}{\beta}$, which means $$F_s(\theta_s^{t+1})\geq F_s(\theta_s^{t})$$
$$\longrightarrow V_\phi^{(t+1)}(s)\geq V_\phi^{(t)}(s);Q_\phi^{(t+1)}(s,a)\geq Q_\phi^{(t)}(s,a).$$
\end{proof}

\begin{lemma}
For all states s and actions a, there exists values $V_\phi^\infty(s), Q_\phi^{\infty}(s,a)\text{ and }Q_\phi^{\infty,i}(s,a)$ such that as $t\rightarrow \infty, V_\phi^t(s)\rightarrow V_\phi^\infty(s),Q_\phi^t(s,a)\rightarrow Q_\phi^\infty(s,a),Q_\phi^{t,i}(s,a)\rightarrow Q_\phi^{\infty,i}(s,a)$.
Define $$\Delta^i = \min_{\{s,a^i|A_\phi^{\infty,i}(s,a^i)\neq 0\}}|A_\phi^{\infty,i}(s,a^i)|.$$
$$\Delta = \min_{i}\Delta^i.$$
Further, there exists a $T_0$ such that $\forall t>T_0, s\in \mathcal{S}, a^i\in \mathcal{A}_\phi^i,$ $$Q_\phi^{\infty,i}(s,a^i)-\frac{\Delta}{4}\leq Q_\phi^{t,i}(s,a^i)\leq Q_\phi^{\infty,i}(s,a^i)+\frac{\Delta}{4}$$
\end{lemma}
\begin{proof}
$\{V_\phi^t(s)\}$ is bounded and monotonically increasing, therefore $V_\phi^t(s)\rightarrow V_\phi^\infty(s)$. Similarly, we know $Q_\phi^t(s,a)\rightarrow Q_\phi^\infty(s,a)$. 
Since the product policy is assumed to converge, we have that $\{Q_\phi^{t,i}(s,a^i)\}$ is convergent. 
For agent $i$, state $s$, categorize the local action $a^i$ into three groups:
$$I_0^{s,i}=\Bigg\{a^i|Q_\phi^{\infty,i}(s,a^i)=V_\phi^\infty(s)\Bigg\}$$
$$I_+^{s,i}=\Bigg\{a^i|Q_\phi^{\infty,i}(s,a^i)>V_\phi^\infty(s)\Bigg\}$$
$$I_-^{s,i}=\Bigg\{a^i|Q_\phi^{\infty,i}(s,a^i)<V_\phi^\infty(s)\Bigg\}$$
Since $Q_\phi^{t,i}(s,a^i)\rightarrow Q_\phi^{\infty,i}(s,a^i)$ as $t\rightarrow \infty$, there exists a $T_0$ such that $\forall t>T_0, s\in \mathcal{S}, a^i\in \mathcal{A}_\phi^i,$ $$Q_\phi^{\infty,i}(s,a^i)-\frac{\Delta}{4}\leq Q_\phi^{t,i}(s,a^i)\leq Q_\phi^{\infty,i}(s,a^i)+\frac{\Delta}{4}$$
\end{proof}

\begin{lemma}
$\exists T_1$ such that $\forall t>T_1,s\in \mathcal{S},$ we have $$A_\phi^{t,i}(s,a^i)<-\frac{\Delta}{4}\text{ for }a^i\in I_-^{s,i};A_\phi^{t,i}(s,a^i)>\frac{\Delta}{4}\text{ for }a^i\in I_+^{s,i}$$
\end{lemma}

\begin{proof}
Since $V_\phi^t(s)\rightarrow V_\phi^\infty(s)$, we have that there exists $T_1>T_0$ such that for all $t>T_1$, $$V_\phi^\infty(s)-\frac{\Delta}{4}\leq V_\phi^t(s)\leq V_\phi^\infty(s)+\frac{\Delta}{4}$$
For $a^i\in I_-^{s,i}, t>T_1>T_0,$  
\begin{equation} \label{eq2}
\begin{split}
A_\phi^{t,i}(s,a^i) & = Q_\phi^{t,i}(s,a^i)-V_\phi^t(s) \\
 & \leq Q_\phi^{\infty,i}(s,a^i)+\frac{\Delta}{4}-V_\phi^t(s)
 \\
 & \leq Q_\phi^{\infty,i}(s,a^i)+\frac{\Delta}{4}-V_\phi^\infty(s)+\frac{\Delta}{4}
 \\
 & \leq -\Delta+\frac{\Delta}{4}+\frac{\Delta}{4}
 \\
 & \leq -\frac{\Delta}{4}
\end{split}
\end{equation}

For $a^i\in I_+^{s,i}, t>T_1>T_0,$  
\begin{equation} \label{eq3}
\begin{split}
A_\phi^{t,i}(s,a^i) & = Q_\phi^{t,i}(s,a^i)-V_\phi^t(s) \\
 & \geq Q_\phi^{\infty,i}(s,a^i)-\frac{\Delta}{4}-V_\phi^t(s)
 \\
 & \geq Q_\phi^{\infty,i}(s,a^i)-\frac{\Delta}{4}-V_\phi^\infty(s)
 \\
 & \geq \Delta-\frac{\Delta}{4}
 \\
 & \geq \frac{\Delta}{4}
\end{split}
\end{equation}
\end{proof}

\begin{lemma}
$\frac{\partial \Phi^{\pi_\theta}(\mu)}{\partial \theta_{s,a^i}^{i}}\rightarrow 0 \text{ as }t\rightarrow\infty$ for all states $s$, agents $i$, actions $a^i$. This implies that $\forall a^i\in I_-^{s,i}\cup I_+^{s,i}, \pi^{t,i}(a^i|s)\rightarrow 0$ and that $\sum_{a^i\in I_0^{s,i}}\pi^{t,i}(a^i|s)\rightarrow 1$.
\end{lemma}

\begin{proof}
Since $\Phi^{\pi_\theta}(\mu)$ is smooth, we know $\frac{\partial \Phi^{\pi_\theta}(\mu)}{\partial \theta_{s,a^i}^{i}}\rightarrow 0$ for all $s,i,a^i$. From lemma 1 we have $$\frac{\partial \Phi^{t}(\mu)}{\partial \theta_{s,a^i}^{i}}=\frac{1}{1-\gamma}d_\mu^{\pi^t}(s)\pi^{t,i}(a^i|s)A_\phi^{\pi^{t,i}}(s,a^i)$$
Since from lemma 4 we know that $|A_\phi^{t,i}(s,a^i)|>\frac{\Delta}{4}$ for all $t>T_1$, for all $a^i\in I_-^{s,i}\cup I_+^{s,i}$
, which together with the assumption that $\mu$ is strict positive for all state $s$ prove $\pi^{t,i}(a^i|s)\rightarrow 0$. Then we also know for all $\sum_{a^i\in I_0^{s,i}}\pi^{t,i}(a^i|s)\rightarrow 1$. 
\end{proof}

\begin{lemma}
For $t\geq T_1$, $\theta_{s,a^i}^{i}$ is strictly decreasing $\forall a^i\in I_-^{s,i}$ and $\theta_{s,a^i}^{i}$ is strictly increasing $\forall a^i\in I+^{s,i}$.
\end{lemma}
\begin{proof}
From lemma 1 we have $$\frac{\partial \Phi^{t}(\mu)}{\partial \theta_{s,a^i}^{i}}=\frac{1}{1-\gamma}d_\mu^{\pi^t}(s)\pi^{t,i}(a^i|s)A_\phi^{{t,i}}(s,a^i)$$
From lemma 4, we know for all $t>T_1,a^i\in I_-^{s,i}, A_\phi^{t,i}(s,a^i) \leq -\frac{\Delta}{4};$ For all $a^i\in I_+^{s,i}, A_\phi^{t,i}(s,a^i) \geq \frac{\Delta}{4}.$ This implies that after iteration $T_1$, $\frac{\partial \Phi^{t}(\mu)}{\partial \theta_{s,a^i}^{i}}<0 \forall a^i\in I_-^{s,i};\frac{\partial \Phi^{t}(\mu)}{\partial \theta_{s,a^i}^{i}}>0 \forall a^i\in I_+^{s,i}.\longrightarrow$ After iteration $T_1$, $\theta_{s,a^i}^{i}$ is strictly decreasing $\forall a^i\in I_-^{s,i}$ and $\theta_{s,a^i}^{i}$ is strictly increasing $\forall a^i\in I_+^{s,i}$.
\end{proof}

\begin{lemma}
For all states $s$ where $I_+^{s,i}\neq \emptyset$, we have:
$$\text{max}_{a^i\in I_0^{s,i}}\theta_{s,a^i}^{t,i}\rightarrow\infty, \text{ min}_{a^i\in \mathbb{A}^i}\theta_{s,a^i}^{t,i}\rightarrow-\infty$$
\end{lemma}

\begin{proof}
Since $I_+^{s,i}\neq \emptyset$, we have some action $a_+^i\in I_+^{s,i}$. From lemma 5, we know $$\pi^{t,i}(a^i_+|s)\rightarrow 0 \text{ as }t\rightarrow \infty$$ 
$$\longrightarrow \frac{\text{exp}(\theta_{s,a_+^i}^{t,i})}{\sum_{a^i\in \mathbb{A}^i}\text{exp}(\theta_{s,a^i}^{t,i})}\rightarrow 0 \text{ as }t\rightarrow \infty$$
From lemma 6 we know $\theta_{s,a_+^i}^{t,i}$ is monotonically increasing, which implies 
$$\sum_{a^i\in \mathbb{A}^i}\text{exp}(\theta_{s,a^i}^{t,i})\rightarrow \infty \text{ as }t\rightarrow \infty$$
From lemma 5, we also know $$\sum_{a^i\in I_0^{s,i}}\pi^{t,i}(a^i|s)\rightarrow 1$$
$$\longrightarrow \frac{\sum_{a^i\in I_0^{s,i}}\text{exp}(\theta_{s,a^i}^{t,i})}{\sum_{a^i\in \mathbb{A}^i}\text{exp}(\theta_{s,a^i}^{t,i})}\rightarrow 1$$
Since denominator does to $\infty$, we know
$$\sum_{a^i\in I_0^{s,i}}\text{exp}(\theta_{s,a^i}^{t,i})\rightarrow\infty$$
which implies $$\text{max}_{a^i\in I_0^{s,i}}\theta_{s,a^i}^{t,i}\rightarrow\infty$$
Note this also implies $\text{max}_{a^i\in \mathbb{A}^i}\theta_{s,a^i}^{t,i}\rightarrow\infty$. The sum of the gradient is always zero: $\sum_{a^i\in \mathbb{A}^i}\frac{\partial \Phi^{t}(\mu)}{\partial \theta_{s,a^i}^{i}}=\frac{1}{1-\gamma}d_\mu^{\pi^t}(s)\sum_{a^i\in \mathbb{A}^i}\pi_{\theta^{i}}(a^i|s)A_\phi^{\pi^{t,i}}(s,a^i)=0$. Thus, $\sum_{a^i\in \mathbb{A}^i}\theta_{s,a^i}^{t,i}=\sum_{a^i\in \mathbb{A}^i}\theta_{s,a^i}^{0,i}$ which is a constant. Since $\text{max}_{a^i\in \mathbb{A}^i}\theta_{s,a^i}^{t,i}\rightarrow\infty$, we know $$\text{min}_{a^i\in \mathbb{A}^i}\theta_{s,a^i}^{t,i}\rightarrow-\infty$$
\end{proof}

\begin{lemma}
Suppose $a_+^i\in I_+^{s,i}$. $\forall a\in I_0^{s,i},$ if $\exists t\geq T_1$ such that $\pi^{t,i}(a|s)\leq \pi^{t,i}(a_+^i|s)$, then $\forall \tau\geq t, \pi^{\tau,i}(a|s)\leq \pi^{\tau,i}(a_+^i|s)$.  

\end{lemma}

\begin{proof}
Suppose $a_+^i\in I_+^{s,i}, a\in I_0^{s,i},$ if $\pi^{t,i}(a|s)\leq \pi^{t,i}(a_+^i|s)$, then
$$\frac{\partial \Phi^{t}(\mu)}{\partial \theta_{s,a}^{i}}=\frac{1}{1-\gamma}d_\mu^{\pi^t}(s)\pi^{t,i}(a|s)(Q_\phi^{{t,i}}(s,a)-V_\Phi^{{t}}(s))$$
$$\leq \frac{1}{1-\gamma}d_\mu^{\pi^t}(s)\pi^{t,i}(a_+^i|s)(Q_\phi^{{t,i}}(s,a_+^i)-V_\Phi^{{t}}(s))=\frac{\partial \Phi^{t}(\mu)}{\partial \theta_{s,a_+^i}^{i}}$$,
where the last step holds because $Q_\phi^{{t,i}}(s,a_+^i)\geq Q_\phi^{{\infty,i}}(s,a_+^i)-\frac{\Delta}{4}\geq Q_\phi^{{\infty,i}}(s,a)+\Delta-\frac{\Delta}{4}\geq Q_\phi^{{t,i}}(s,a)-\frac{\Delta}{4}+\Delta-\frac{\Delta}{4}> Q_\phi^{{t,i}}(s,a)$
for $t>T_0$.
\\We can then partition $I_0^{s,i}$ into $B_0^{s,i}(a_+^i)$ and $\Bar{B}_0^{s,i}(a_+^i)$ as follows: 
$$B_0^{s,i}(a_+^i):\{a|a\in I_0^{s,i}\text{ and }\forall t\geq T_0,\pi^{t,i}(a_+^i|s)< \pi^{t,i}(a|s)\}$$
$$\Bar{B}_0^{s,i}(a_+^i):I_0^{s,i}\setminus B_0^{s,i}(a_+^i).$$
\end{proof}

\begin{lemma}
Suppose $I_+^{s,i}\neq\emptyset$. $\forall a_+^i\in I_+^{s,i}$, we have that $B_0^{s,i}(a_+^i)\neq\emptyset$ and that $$\sum_{a^i\in B_0^{s,i}(a_+^i)}\pi^{t,i}(a^i|s)\rightarrow 1\text{, as } t\rightarrow\infty.$$
This implies that: $$\text{max}_{a^i\in B_0^{s,i}(a_+^i)}\theta_{s}^{t,i}\rightarrow\infty.$$

\end{lemma}

\begin{proof}
Let $a_+^i\in I_+^{s,i}$. Consider any $\Bar{a}^i\in \Bar{B}_0^{s,i}(a_+^i)$. Then by definition of $\Bar{B}_0^{s,i}(a_+^i)$, there exists $t^\prime>T_0$ such that $\pi^{t,i}(a_+^i|s)\geq \pi^{t,i}(\Bar{a}^i|s)$. From lemma 8, we know $\forall \tau>t,\pi^{\tau,i}(a_+^i|s)\geq \pi^{\tau,i}(\Bar{a}^i|s)$. From lemma 5, we know $\pi^{t,i}(a_+^i|s)\rightarrow 0 \text{ as }t\rightarrow \infty$, which implies $$\pi^{t,i}(\Bar{a}^i|s)\rightarrow 0 \text{ as }t\rightarrow \infty.$$ \\Since $B_0^{s,i}(a_+^i)\cup \Bar{B}_0^{s,i}(a_+^i) = I_0^{s,i}$ and $\sum_{a^i\in I_0^{s,i}}\pi^{t,i}(a^i|s)\rightarrow 1$, we know $$\sum_{a^i\in B_0^{s,i}(a_+^i)}\pi^{t,i}(a^i|s)\rightarrow 1$$
$$B_0^{s,i}(a_+^i)\neq \oldemptyset$$
Using the same techniques in lemma 7, we know $$\text{max}_{a^i\in B_0^{s,i}(a_+^i)}\theta_{s}^{t,i}\rightarrow\infty$$
\end{proof}

\begin{lemma}
Consider any $s$ where $I_+^{s,i}\neq\emptyset$. Then, $\forall a_+^i\in I_+^{s,i}, \exists T_{a_+^i}$ such that $\forall t>T_{a_+^i},\forall a^i\in\Bar{B}_0^{s,i}(a_+^i),$
$$\pi^{t,i}(a_+^i|s)> \pi^{t,i}(a^i|s)$$
\end{lemma}

\begin{proof}
By the definition of $\Bar{B}_0^{s,i}(a_+^i)$ and lemma 8, there exists $t_{a^i}>T_0$ such that $\forall \tau>t_{a^i}$, $\pi^{\tau,i}(a_+^i|s)> \pi^{\tau,i}(a^i|s)$. We can choose $T_{a_+^i}=\text{max}_{a^i\in B_0^{s,i}(a_+^i)}t_{a^i}$.
\end{proof}

\begin{lemma}
$\forall a^i_+\in I_+^{s,i},$ we have $\theta_{s,a^i_+}^{i}$ is lower bounded as $t\rightarrow\infty$. $\forall a^i_-\in I_-^{s,i},$ we have that $\theta_{s,a^i_-}^{i}\rightarrow-\infty$ as $t\rightarrow\infty$. 
\end{lemma}

\begin{proof}
From lemma 6, we know that $\forall a^i_+\in I_+^{s,i}$, after $T_1$, $\theta_{s,a^i_+}^{i}$ is strictly increasing, and is therefore bounded from below.
\\For the second claim, we know from lemma 6 that $\forall a^i_-\in I_-^{s,i}$, after $T_1$, 
$\theta_{s,a^i_-}^{i}$ is strictly decreasing. Then, by monotone convergence theorem, we know $\text{lim}_{t\rightarrow \infty}\theta_{s,a^i_-}^{i}$ exists and is either $-\infty$ or some constant $\theta_{0}^{i}$. We now prove by contraction that $\text{lim}_{t\rightarrow \infty}\theta_{s,a^i_-}^{i}$ cannot be some constant $\theta_{0}^{i}$. Suppose $\text{lim}_{t\rightarrow \infty}\theta_{s,a^i_-}^{i}=\theta_{0}^{i}$. We immediately know that $\forall t\geq T_1,\theta_{s,a^i_-}^{i}>\theta_{0}^{i}$. By lemma 7, we know $\exists a^i \in \mathbb{A_\phi}^i$ such that \begin{equation}\label{eq:3}\text{lim }\underset{t\rightarrow \infty}{\text{inf}}\theta_{s,a^i}^{t,i}=-\infty\end{equation}
Let us consider some $\delta^i>0$ such that $\theta_{s,a^i}^{T_1,i}\geq \theta_{0}^{i}-\delta^i$. Now for $t\geq T_1$, define $\tau^i(t)$ to be the largest iteration in $[T_1,t]$ such that $\theta_{s,a^i}^{\tau^i(t),i}\geq \theta_{0}^{i}-\delta^i$.
Define $\Tau^{t,i}$ to be the subsequence $\{t^{\prime}\}$ of the interval $(\tau^i(t),t)$ such that $\theta_{s,a^i}^{t^\prime,i}$ decreases. 
\\Define 
$$Z^{t,i}=\sum_{t^\prime\in\Tau^{t,i}}\frac{\partial \Phi^{t^\prime}(\mu)}{\partial \theta_{s,a^i}^{i}}$$
For non-empty $\Tau^{t,i}$, we have:
$$Z^{t,i}=\sum_{t^\prime\in\Tau^{t,i}}\frac{\partial \Phi^{t^\prime}(\mu)}{\partial \theta_{s,a^i}^{i}}\leq \sum_{t^\prime=\tau^i(t)+1}^{t-1}\frac{\partial \Phi^{t^\prime}(\mu)}{\partial \theta_{s,a^i}^{i}}\leq \sum_{t^\prime=\tau^i(t)}^{t-1}\frac{\partial \Phi^{t^\prime}(\mu)}{\partial \theta_{s,a^i}^{i}}+\frac{1}{(1-\gamma)}(\Phi_{\text{max}}-\Phi_{\text{min}})$$
$$=\frac{1}{\eta}(\theta_{s,a^i}^{t,i}-\theta_{s,a^i}^{\tau^i(t),i})+\frac{1}{(1-\gamma)}(\Phi_{\text{max}}-\Phi_{\text{min}})$$
where we have used that $|\frac{\partial \Phi^{t^\prime}(\mu)}{\partial \theta_{s,a^i}^{i}}|\leq \frac{1}{(1-\gamma)}(\Phi_{\text{max}}-\Phi_{\text{min}})$. \\By equation $(\ref{eq:3})$, we know \begin{equation} \label{eq:4}\text{lim }\underset{t\rightarrow \infty}{\text{inf}}Z^{t,i}=-\infty\end{equation}
\\For any $\Tau^{t,i}\neq \emptyset,\forall t^{\prime}\in \Tau^{t,i}$, from lemma 1, we know:
$$\Bigg|\frac{\partial \Phi^{t^\prime}(\mu)/\partial \theta_{s,a^i_-}^{i}}{\partial \Phi^{t^\prime}(\mu)/\partial \theta_{s,a^i}^{i}}\Bigg|=\Bigg|\frac{\pi^{t^\prime,i}(a_-^i|s)A_\phi^{t^\prime,i}(s,a_-^i)}{\pi^{t^\prime,i}(a^i|s)A_\phi^{t^\prime,i}(s,a^i)}\Bigg|\geq \text{exp}(\theta_{0}^{i}-\theta_{s,a^i}^{t^\prime,i})\frac{\Delta}{4(\Phi_{\text{max}}-\Phi_{\text{min}})}$$
$$\geq \text{exp}(\delta^i)\frac{\Delta}{4(\Phi_{\text{max}}-\Phi_{\text{min}})}$$
where we have used that $|A_\phi^{t^\prime,i}(s,a^i)|\leq \Phi_{\text{max}}-\Phi_{\text{min}}$ and $\forall t^\prime>T_1,|A_\phi^{t^\prime,i}(s,a_-^i)|\geq \frac{\Delta}{4}$.
Since both $\frac{\partial \Phi^{t^\prime}(\mu)}{\partial \theta_{s,a^i_-}^{i}}$ and $\frac{\partial \Phi^{t^\prime}(\mu)}{\partial \theta_{s,a^i}^{i}}$ are negative, we can get: \begin{equation} \label{eq:5}
\frac{\partial \Phi^{t^\prime}(\mu)}{\partial \theta_{s,a^i_-}^{i}}\leq \text{exp}(\delta^i)\frac{\Delta}{4(\Phi_{\text{max}}-\Phi_{\text{min}})}\frac{\partial \Phi^{t^\prime}(\mu)}{\partial \theta_{s,a^i}^{i}} 
\end{equation}
For non-empty $\Tau^{t,i}$, 
$$\frac{1}{\eta}(\theta_{s,a^i_-}^{t,i}-\theta_{s,a^i_-}^{T_1,i})=\sum_{t^\prime=T_1}^{t-1}\frac{\partial \Phi^{t^\prime}(\mu)}{\partial \theta_{s,a^i_-}^{i}}\leq \sum_{t^\prime\in\Tau^{t,i}}\frac{\partial \Phi^{t^\prime}(\mu)}{\partial \theta_{s,a^i_-}^{i}}$$
By equation ($\ref{eq:5}$)
$$\leq \text{exp}(\delta^i)\frac{\Delta}{4(\Phi_{\text{max}}-\Phi_{\text{min}})}\sum_{t^\prime\in\Tau^{t,i}}\frac{\partial \Phi^{t^\prime}(\mu)}{\partial \theta_{s,a^i}^{i}}$$
$$=\text{exp}(\delta^i)\frac{\Delta}{4(\Phi_{\text{max}}-\Phi_{\text{min}})}Z^{t,i}$$
which together with the fact that $\theta_{s,a^i_-}^{T_1,i}$ is some finite constant and equation ($\ref{eq:4}$) lead to $$\theta_{s,a^i_-}^{t,i}\rightarrow -\infty \text{ as }t\rightarrow\infty$$
this contradicts the assumption that $\{\theta_{s,a_-^{t,i}}^i\}_{t\geq T_1}$ is lower bounded by $\theta_0^i$ and complete the proof.
\end{proof}

\begin{lemma}
Consider any $s$ where $I_+^{s,i}\neq \emptyset$. Then, $\forall a_+^i\in I_+^{s,i}$, $$\sum_{a^i\in B_0^{s,i}(a_+^i)}\theta_{s,a}^{t,i}\rightarrow\infty $$

\end{lemma}

\begin{proof}
For any $a^i\in B_0^{s,i}(a_+^i)$. By definition, we know that $\forall t>T_0, \pi^{t,i}(a_+^i|s)< \pi^{t,i}(a|s)$, which implies that $\theta_{s,a^i_+}^{t,i}<\theta_{s,a}^{t,i}$. Since in lemma 11, $\theta_{s,a^i_+}^{t,i}$ is lower bounded as $t\rightarrow\infty$, we know that $\theta_{s,a}^{t,i}$ is lower bounded as $t\rightarrow\infty$. This together with lemma 9 proves that $$\sum_{a^i\in B_0^{s,i}(a_+^i)}\theta_{s,a}^{t,i}\rightarrow\infty $$
\end{proof}

\begin{proof} [Proof of Theorem \ref{theorem:Asymptotic convergence to Nash with gradient ascent}]
Suppose $I_+^{s,i}$ is non-empty for some $s$, else the proof is complete. Let $a_+^i\in I_+^{s,i}$. Then, by lemma 12, we know
\begin{equation} \label{eq:6}
\sum_{a^i\in B_0^{s,i}(a_+^i)}\theta_{s}^{t,i}\rightarrow\infty
\end{equation}
For $a^i\in I_-^{s,i}$, since $\frac{\pi^{t,i}(a^i|s)}{\pi^{t,i}(a_+^i|s)}=\text{exp}(\theta_{s}^{t,i}-\theta_{s,a^i_+}^{t,i})\rightarrow 0$ (as $\theta_{s,a^i_+}^{t,i}$ is lower bounded and $\theta_{s}^{t,i}\rightarrow-\infty$ by lemma 11), there exists $T_2>T_0$ such that $$\frac{\pi^{t,i}(a^i|s)}{\pi^{t,i}(a_+^i|s)}<\frac{\Delta}{8|\mathcal{A}_\phi^i|(\Phi_{\text{max}}-\Phi_{\text{min}})}$$
\begin{equation} \label{eq:7}
\longrightarrow -\sum_{a^i\in I_-^{s,i}}\frac{\pi^{t,i}(a^i|s)}{\Phi_{\text{max}}-\Phi_{\text{min}}}>-\pi^{t,i}(a_+^i|s)\frac{\Delta}{8}
\end{equation}
For $a^i\in \Bar{B}_0^{s,i}$, by definition of $\Bar{B}_0^{s,i}$, we have $A_\phi^{t,i}(s,a^i)\rightarrow 0$ and by lemma 10, $\forall t>T_{a_+^i} 1<\frac{\pi^{t,i}(a_+^i|s)}{\pi^{t,i}(a^i|s)}$
. Then, $\exists T_3>T_2,T_{a_+^i}$ such that 
$$|A_\phi^{t,i}(s,a^i)|<\frac{\pi^{t,i}(a_+^i|s)}{\pi^{t,i}(a^i|s)}\frac{\Delta}{16|\mathcal{A}_\phi^i|}$$
$$\longrightarrow\sum_{a^i\in \Bar{B}_0^{s,i}(a_+^i)}\pi^{t,i}(a^i|s)|A_\phi^{t,i}(s,a^i)|<\pi^{t,i}(a^i_+|s)\frac{\Delta}{16}$$
\begin{equation} \label{eq:8}
\longrightarrow -\pi^{t,i}(a^i_+|s)\frac{\Delta}{16}<\sum_{a^i\in \Bar{B}_0^{s,i}(a_+^i)}\pi^{t,i}(a^i|s)A_\phi^{t,i}(s,a^i)<\pi^{t,i}(a^i_+|s)\frac{\Delta}{16}
\end{equation}
For $t>T_3$, 
$$0=\sum_{a^i\in\mathcal{A}_\phi^i}\pi^{t,i}(a^i|s)A_\phi^{t,i}(s,a^i)$$
$$=\sum_{a^i\in I_0^{s,i}}\pi^{t,i}(a^i|s)A_\phi^{t,i}(s,a^i)+\sum_{a^i\in I_+^{s,i}}\pi^{t,i}(a^i|s)A_\phi^{t,i}(s,a^i)+\sum_{a^i\in I_-^{s,i}}\pi^{t,i}(a^i|s)A_\phi^{t,i}(s,a^i)$$
$$\stackrel{(a)}{\geq}\sum_{a^i\in B_0^{s,i}(a_+^i)}\pi^{t,i}(a^i|s)A_\phi^{t,i}(s,a^i)+\sum_{a^i\in \Bar{B}_0^{s,i}(a_+^i)}\pi^{t,i}(a^i|s)A_\phi^{t,i}(s,a^i)$$
$$+\pi^{t,i}(a^i_+|s)A_\phi^{t,i}(s,a^i_+)+\sum_{a^i\in I_-^{s,i}}\pi^{t,i}(a^i|s)A_\phi^{t,i}(s,a^i)$$
$$\stackrel{(b)}{\geq}\sum_{a^i\in B_0^{s,i}(a_+^i)}\pi^{t,i}(a^i|s)A_\phi^{t,i}(s,a^i)+\sum_{a^i\in \Bar{B}_0^{s,i}(a_+^i)}\pi^{t,i}(a^i|s)A_\phi^{t,i}(s,a^i)+\pi^{t,i}(a^i_+|s)\frac{\Delta}{4}-\sum_{a^i\in I_-^{s,i}}\frac{\pi^{t,i}(a^i|s)}{\Phi_{\text{max}}-\Phi_{\text{min}}}$$
$$\stackrel{(c)}{\geq}\sum_{a^i\in B_0^{s,i}(a_+^i)}\pi^{t,i}(a^i|s)A_\phi^{t,i}(s,a^i)-\pi^{t,i}(a^i_+|s)\frac{\Delta}{16}+\pi^{t,i}(a^i_+|s)\frac{\Delta}{4}-\pi^{t,i}(a_+^i|s)\frac{\Delta}{8}$$
$$>\sum_{a^i\in B_0^{s,i}(a_+^i)}\pi^{t,i}(a^i|s)A_\phi^{t,i}(s,a^i)$$
where (a) uses $\forall a^i\in I_+^{s,i} \text{ and }t>T_3>T_1, A_\phi^{t,i}(s,a^i)>0$ from lemma 3, (b) uses $\forall t>T_3>T_1, A_\phi^{t,i}(s,a^i_+)>\frac{\Delta}{4}$ from lemma 3 and $A_\phi^{t,i}(s,a^i)\geq -(\Phi_{\text{max}}-\Phi_{\text{min}})$, (c) uses equation $(\ref{eq:7})$ and equation $(\ref{eq:8})$. This implies that $$\forall t>T_3, \sum_{a^i\in B_0^{s,i}(a_+^i)}\frac{\partial \Phi^{t}(\mu)}{\partial \theta_{s,a}^{i}}<0$$
which contradicts with equation $(\ref{eq:6})$ which leads to $$\text{lim}_{t\rightarrow\infty}\sum_{a^i\in B_0^{s,i}(a_+^i)}(\theta_{s,a^i}^{t,i}-\theta_{s,a^i}^{T_3,i})=\eta\sum_{t=T_3}^\infty\sum_{a^i\in B_0^{s,i}(a_+^i)}\frac{\partial \Phi^{t}(\mu)}{\partial \theta_{s,a}^{i}}\rightarrow\infty$$
Therefore, the set $I_+^{s,i}=\emptyset$.
\\Let $\theta=[\theta^{i,\infty},\theta^{-i,\infty}],\theta^\prime=[\theta^{i},\theta^{-i,\infty}]$.
$$V^{\pi_{\theta^{\prime}}}(\mu)-V^{\pi_{\theta}}(\mu)=\Phi^{\pi_{\theta^{\prime}}}(\mu)-\Phi^{\pi_{\theta^i}}(\mu)$$
$$=\E_{s_0\sim\mu}[V_\Phi^{\pi_{\theta^{\prime}}}(s_0)-V_\Phi^{\pi_{\theta}}(s_0)]$$
By performance difference lemma,
$$=\frac{1}{1-\gamma}\E_{s\sim d_\mu^{\pi_{\theta^{\prime}}}}[\E_{a\sim \pi_{\theta^{\prime}}(\cdot|s)}A^{\pi_\theta}_\Phi(s,a)]$$
$$=\frac{1}{1-\gamma}\E_{s\sim d_\mu^{\pi_{\theta^{\prime}}}}[\E_{a^i\sim \pi^i(\cdot|s)}[\E_{a^{-i}\sim\pi_{\theta^{\infty,-i}}(\cdot|s)}A^{\pi_\theta}_\Phi(s,a)]]$$
$$=\frac{1}{1-\gamma}\E_{s\sim d_\mu^{\pi_{\theta^{\prime}}}}[\E_{a^i\sim \pi^{i}(\cdot|s)}A^{\infty,i}_\Phi(s,a^i)]$$
Since $I_+^{s,i}=\emptyset$,
$$\leq\frac{1}{1-\gamma}\E_{s\sim d_\mu^{\pi_{\theta^{\prime}}}}[\E_{a^i\sim \pi^{\infty,i}(\cdot|s)}A^{\infty,i}_\Phi(s,a^i)]$$
$$=0$$
which completes the proof.

\end{proof}

\section{Proofs for Section \ref{sec:Policy gradient dynamics with log-barrier regularization}}

\subsection{Proof of Lemma \ref{lemma:Log barrier regularization's approximate first-order stationary points are near-Nash}}
\label{sec:Proof of Lemma lemma:Log barrier regularization's approximate first-order stationary points are near-Nash}
The proof extends the proof of Theorem 5.2 in \cite{agarwal2019theory} by the usage of the multi-agent performance difference lemma (Lemma C.1 in \cite{leonardos2021global}).

Fix an arbitrary agent $i\in\mathcal{N}$ and suppose it deviates from $\pi^i_{\theta^i}$ to an optimal policy $\pi^i_*(\theta^{-i})$ w.r.t.~the corresponding single-agent MDP specified by $\theta^{-i}$.
We will use $\pi^i_*$ as a shorthand for $\pi^i_*(\theta^{-i})$ and $\pi^{-i}$ as a shorthand for $\pi^{-i}_{\theta^{-i}}$.
By the definition of $\epsilon$-Nash, we need to show that $V^i_{\pi^i_*,\pi^{-i}}(\mu)-V^i_\theta(\mu)\leq2\lambda M$.

Similar to the proof of Theorem 5.2 in \cite{agarwal2019theory}, we can bound $A^i_\theta(s,a^i)\leq$ for any $(s,a^i)$-pair.
It suffices to bound $A^i_\theta(s,a^i)$ for any $(s,a^i)$ where $A^i_\theta(s,a^i) \geq 0$ (else $A^i_{\theta}(s,a^i)\leq$ is trivially true):
\begin{align*}
    \lambda/(2|\mathcal{S}||\mathcal{A}^i|)=: \epsilon_{\rm opt} \geq \frac{\partial L_\lambda(\theta)}{\partial \theta^i_{s,a^i}} \stackrel{{\rm (i)}}{=} d^{\pi_\theta}_\mu(s) \pi^i_{\theta^i}(a^i|s)A^i_\theta(s, a^i) + \frac{\lambda}{|\mathcal{S}|}\left(\frac{1}{|\mathcal{A}^i|} - \pi^i_{\theta^i}(a^i|s)\right)
    \geq \frac{\lambda}{|\mathcal{S}|}\left(\frac{1}{|\mathcal{A}^i|} - \pi^i_{\theta^i}(a^i|s)\right)
\end{align*}
where the last inequality is due to $A^i_\theta(s,a^i) \geq 0$,
and by rearranging we get $\pi^i_{\theta^i}(a^i|s)\geq 1/{2|\mathcal{A}^i|}$.
Solving {\rm (i)} for $A^i_\theta(s, a^i)$, we have
\begin{align*}
    A^i_\theta(s, a^i) =& \frac{1}{d^{\pi_\theta}_\mu(s)}\left(\frac{1}{\pi^i_{\theta^i}(a^i|s)}\frac{\partial L_\lambda(\theta)}{\partial \theta^i_{s,a^i}}+\frac{\lambda}{|\mathcal{S}|}\left(1-\frac{1}{\pi^i_{\theta^i}(a^i|s)|\mathcal{A}^i|}\right)\right)\\
    \leq& \frac{1}{d^{\pi_\theta}_\mu(s)}\left(2|\mathcal{A}^i|\epsilon_{\rm opt}+\frac{\lambda}{|\mathcal{S}|}\right) \qquad \text{($\pi^i_{\theta^i}(a^i|s)\geq 1/{2|\mathcal{A}^i|}$)}\\
    \leq& \frac{2\lambda}{{d^{\pi_\theta}_\mu(s)}|\mathcal{S}|} \qquad 
    \text{($\epsilon_{\rm opt}=\lambda/(2|\mathcal{S}||\mathcal{A}^i|)$)}
\end{align*}
We are now ready to use the multi-agent performance difference lemma on $\pi_*:=(\pi^i_*,\pi^{-i})$ and $\pi_\theta$:
\begin{align*}
    V^i_{\pi^i_*,\pi^{-i}}(\mu)-V^i_\theta(\mu) =& \E_{s\sim d^{\pi_*}_\mu}\E_{a^i\sim \pi^i_*(s)}\E_{a^{-i}\sim\pi^{-i}}\left[A^i_\theta(s,a^i,a^{-i})\right]\\
    =& \sum_s d^{\pi_*}_\mu(s)\sum_{a^i}\pi^i_*(a^i|s)A^i_\theta(s,a^i)\\
    \leq& \sum_s d^{\pi_*}_\mu(s) \frac{2
    \lambda}{d^{\pi_\theta}_\mu(s)|\mathcal{S}|} \leq 2\lambda\max_s \left(\frac{d^{\pi_*}_\mu(s)}{d^{\pi_\theta}_\mu(s)}\right)\leq 2\lambda M
\end{align*}
which concludes the proof.

\subsection{Proof of Theorem \ref{theorem:Convergence rate with log barrier regularization}}
\label{sec:Proof of Theorem theorem:Convergence rate with log barrier regularization}
Lemma \ref{lemma:smoothness_tabular_softmax} shows that $\Phi_\theta$ is $\frac{41N}{4(1-\gamma)^3}$-smooth.
Lemma D.4 in \cite{agarwal2019theory} shows that the regularizer for each agent $i$ is $\frac{2\lambda}{|\mathcal{S}|}$-smooth.
Thus, $\beta_\lambda$ is an upper bound on the smoothness of $L_\lambda(\theta)$.
Then, by standard results, we have
\begin{align*}
    \min_{t\leq T}\norm{\nabla_\theta L_\lambda(\theta^{(t)})}^2_2\leq\frac{2\beta_\lambda(L_\lambda(\theta^*)-L_\lambda(\theta_0))}{T} \leq \frac{2\beta_\lambda(\Phi_{\rm max}-\Phi_{\rm min})}{T}
    ,
\end{align*}
where the last inequality is because.
We need to choose $T$ large enough such that 
\begin{align*} \sqrt{\frac{2\beta_\lambda(\Phi_{\rm max}-\Phi_{\rm min})}{T}} \leq\lambda/(2|\mathcal{S}|\max_i|\mathcal{A}^i|)
    .
\end{align*}
Solving the above inequality we obtain $T\geq\frac{8\beta_\lambda|\mathcal{S}|^2\max_i|\mathcal{A}^i|^2(\Phi_{\rm max}-\Phi_{\rm min})}{\lambda^2}$.
By Lemma \ref{lemma:Log barrier regularization's approximate first-order stationary points are near-Nash}, we should set $\lambda=\epsilon/2M$ to achieve the specified Nash-gap of $\epsilon$.
Plugging in $\lambda=\epsilon/2M$ and $\beta_\lambda:=\frac{41N}{4(1-\gamma)^3}+\frac{2\lambda N}{|\mathcal{S}|}$, we have
\begin{align*}
    T \geq& \frac{32M^2|\mathcal{S}|^2\max_i|\mathcal{A}^i|^2\beta_\lambda(\Phi_{\rm max}-\Phi_{\rm min})}{\epsilon^2}\\
    =& \frac{328NM^2|\mathcal{S}|^2\max_i|\mathcal{A}^i|^2(\Phi_{\rm max}-\Phi_{\rm min})}{(1-\gamma)^3\epsilon^2} + \frac{64\lambda NM^2|\mathcal{S}|\max_i|\mathcal{A}^i|^2(\Phi_{\rm max}-\Phi_{\rm min})}{\epsilon^2} \\
    =&\frac{328NM^2|\mathcal{S}|^2\max_i|\mathcal{A}^i|^2(\Phi_{\rm max}-\Phi_{\rm min})}{(1-\gamma)^3\epsilon^2} + \frac{32 NM|\mathcal{S}|\max_i|\mathcal{A}^i|^2(\Phi_{\rm max}-\Phi_{\rm min})}{\epsilon}
\end{align*}
which completes the proof.

\section{Proofs for Section \ref{sec:Approximate best-response natural policy gradient dynamics}}

\subsection{Proof of Lemma \ref{lemma:NPG is effectively soft policy iteration}}
\label{sec:Proof of Lemma lemma:NPG is effectively soft policy iteration}
The proof is similar to that of the counterpart lemma for the single-agent setting (Lemma 5.1 of \cite{agarwal2019theory}).

For a vector $w\in\R^{|\mathcal{S}||\mathcal{A}^i|}$, define the error function
\begin{align*}
    L^i_\theta(w) = \E_{s\sim d^{\pi_\theta}_\mu, a^i \sim \pi^i_{\theta^i}(\cdot|s)}\left[w^\top\nabla_{\theta^i}\log\pi^i_{\theta^i}(a^i|s)-A^i_\theta(s,a^i)\right] = \norm{D^i_\theta \left((\nabla_{\theta^i}\log\pi^i_{\theta^i})w-A^i_\theta\right)}_2^2
\end{align*}
where 
$D^i_\theta\in\R^{|\mathcal{S}||\mathcal{A}^i| \times |\mathcal{S}||\mathcal{A}^i|}$ is the diagonal matrix with diagonal entries $\{d^{\pi_\theta}_\mu(s) \pi^i_{\theta^i}(a^i|s)\}_{s, a^i}$, and
$\nabla_{\theta^i}\log\pi^i_{\theta^i} \in \R^{|\mathcal{S}||\mathcal{A}^i| \times |\mathcal{S}||\mathcal{A}^i|}$ is the Jacobian matrix.
By the main property of the Moore–Penrose inverse for least squares, i.e., the minimizer of $\norm{Ax-b}_2^2$ with the smallest $\ell_2$ norm is $A^\dagger b$, we have
\begin{align*}
    w^*_\theta = \left(D^i_\theta(\nabla_{\theta^i}\log\pi^i_{\theta^i})\right)^\dagger\left(D^i_\theta A^i_\theta\right)
\end{align*}
where $w^*_\theta$ is the minimizer of $L^i_\theta(w)$ with the smallest $\ell_2$ norm.
One can verify that $w^*_\theta = $:
\begin{align*}
    (F^i_{\theta})^\dagger \nabla_{\theta^i}V^i_\theta(\mu) =& \left((\nabla_{\theta^i}\log\pi^i_{\theta^i})^\top D^i_\theta \nabla_{\theta^i}\log\pi^i_{\theta^i}\right)^\dagger \left((\nabla_{\theta^i}\log\pi^i_{\theta^i})^\top D^i_\theta A^i_\theta\right)\\
    =& \left(D^i_\theta \nabla_{\theta^i}\log\pi^i_{\theta^i}\right)^\dagger \left((\nabla_{\theta^i}\log\pi^i_{\theta^i})^\top \right)^\dagger \left((\nabla_{\theta^i}\log\pi^i_{\theta^i})^\top D^i_\theta A^i_\theta\right)\\
    =& \left(D^i_\theta \nabla_{\theta^i}\log\pi^i_{\theta^i}\right)^\dagger  \left(D^i_\theta A^i_\theta\right) \\
    =&  w^*_\theta
\end{align*}
We can then follow the same argument in the proof of Lemma 5.1 in \cite{agarwal2019theory} to show the claim of Lemma \ref{lemma:NPG is effectively soft policy iteration}.

\subsection{Proof of Theorem \ref{theorem:Convergence of approximate-best-response NPG}}
\label{sec:Proof of Theorem theorem:Convergence of approximate-best-response NPG}
Suppose the inner loop achieves $\frac{\epsilon}{2}$-near-optimal deviation, which require at most $\frac{4}{(1-\gamma)^2 \epsilon}$ inner iterations \cite{agarwal2019theory}.
Then, either the best-response iteration halts, or the total potential function is improved by at least $\frac{\epsilon}{2}$, which implies the number of outer iterations is at most $O(\frac{1}{(1-\gamma) \epsilon})$.

\section{Proofs for Section \ref{sec:Bounding the price of anarchy in smooth Markov (potential) games}}

\subsection{Proof of Theorem \ref{theorem:POA bound of maximum-gain epsilon-ratio-best-response in smooth MPGs}}
\label{sec:Proof of Theorem theorem:POA bound of maximum-gain epsilon-ratio-best-response in smooth MPGs}
We abbreviate $V^i_\pi(\mu)$ as $V^i_\pi$ and $V_\pi(\mu)$ as $V_\pi$. For any policy $\pi_t$, define $\delta^i(\pi_t) := V^i_{\pi^{i}_*,\pi^{-i}_t} - V^i_{\pi_t}$ and $\Delta(\pi_t) := \sum_i\delta^i(\pi_t)$.
We now have
\begin{align*}
    V_{\pi_t} = \sum_i V^i_{\pi_t} = \sum_i \left(V^i_{\pi^{i}_*,\pi^{-i}_t}-\delta^i(\pi_t)\right)
    \geq \alpha V_{\pi_*} - \beta V_{\pi_t} - \Delta(\pi_t)
\end{align*}
where the inequality is due to the $(\alpha,\beta)$-smoothness of the MPG, which implies
\begin{align} \label{eq:lower bound V_pi_t}
    V_{\pi_t} \geq \frac{\alpha}{1+\beta} V_{\pi_*} - \frac{1}{1+\beta}\Delta(\pi_t).
\end{align}

For a ``bad'' policy $\pi_t$ that violates \eqref{eq:good policies}, we have 
\begin{align*}
    \Delta(\pi_t) \geq& \alpha V_{\pi_*} - (1+\beta)V_{\pi_t}
    >  (1+\beta)(1+\sigma)V_{\pi_t} - (1+\beta) V_{\pi_t} = \sigma(1+\beta)V_{\pi_t} \geq \sigma(1+\beta)\Phi_{\pi_t}
\end{align*}
where the first inequality is directly from inequality \eqref{eq:lower bound V_pi_t}, the second inequality due to that $\pi_t$ is a bad policy, the third due to the assumption that $\Phi_\pi(s) \leq V_\pi(s)$.
Therefore, for the maximum-gain agent chosen to update from $t$ to $t+1$, the increase in its local value is at least $\frac{\sigma(1+\beta)}{N}\Phi_{\pi_t}$ since $\Delta(\pi_t) = \sum_i\delta^i(\pi_t)$. 
Due to the characteristic of $\Phi$ in \eqref{eq:total potential function}, we have
$\Phi_{t+1} -\Phi_{\pi_t} \geq \frac{\sigma(1+\beta)}{N}\Phi_{\pi_t}$, i.e.,
\begin{align}\label{eq:Phi increase ratio}
    \Phi_{t+1} \geq \left(1+{\sigma(1+\beta)}/{N}\right) \Phi_{\pi_t}.
\end{align}

For a good $\pi_t$ being updated, $\Phi$ can increase by a ratio of at least $\frac{1}{1-\epsilon}$ since 
\begin{align*}
   \frac{\Phi_{t+1}-\Phi_{t}}{\Phi_{t}} = \frac{V^i_{\pi_{t+1}}-V^i_{\pi_{t}}}{\Phi_{t}} = \frac{V_{\pi_{t+1}}-V_{\pi_{t}}}{\Phi_{t}} \geq
   \frac{V_{\pi_{t+1}}-V_{\pi_{t}}}{V_{\pi_{t}}} > \frac{1}{1-\epsilon} - 1
   .
\end{align*}
Let $m$ and $T-m$ be the number of bad and good policies in the sequence, respectively.
We then have $\Phi_0 (1+\frac{\sigma(1+\beta)}{N})^m(\frac{1}{1-\epsilon})^{T-m} \leq \Phi_{\rm max}$, which implies \eqref{eq:bound the number of bad, epsilon-ratio} and concludes the proof.

\subsection{Proof of Corollary \ref{corollary:Bounding the POA of NPG-BR dynamics in smooth MPGs}}
\label{sec:Proof of Corollary corollary:Bounding the POA of NPG-BR dynamics in smooth MPGs}
Similar to the proof of Theorem \ref{theorem:POA bound of maximum-gain epsilon-ratio-best-response in smooth MPGs}, we can obtain inequality \eqref{eq:Phi increase ratio} for a bad $\pi_t$, and for a good $\pi_t$, the $\epsilon/2$ increase per iteration implies
\begin{align*}
   \frac{\Phi_{t+1}-\Phi_{t}}{\Phi_{t}} = \frac{V^i_{\pi_{t+1}}-V^i_{\pi_{t}}}{\Phi_{t}} = \frac{V_{\pi_{t+1}}-V_{\pi_{t}}}{\Phi_{t}} \geq
   \frac{V_{\pi_{t+1}}-V_{\pi_{t}}}{V_{\pi_{t}}} > \frac{\epsilon/2}{1-\gamma}
   .
\end{align*}
Let $m$ and $T-m$ be the number of bad and good policies in the sequence, respectively.
We then have $\Phi_0 (1+\frac{\sigma(1+\beta)}{N})^m(1+\frac{\epsilon}{2(1-\gamma)})^{T-m} \leq \Phi_{\rm max}$, which implies \eqref{eq:bound the number of bad, NPG-BR} and concludes the proof.

\newpage\clearpage
\section{Experiment details}
\label{sec:Experiment details}
\subsection{Pseudocode for the reward function of our Coordination Game}
\begin{algorithm}
	\caption{Calculate the team reward for $N$ agents in state $s$}
	\begin{algorithmic}
	    \IF {$(N=2)$ or $(N=3)$}
		\STATE \texttt{difference\_bound=1}
        \ELSE 
        \STATE \texttt{difference\_bound=2}
		\ENDIF
		
		\IF {$\text{abs}(s.\text{count("0")}-s.\text{count("1"))}\leq \text{difference\_bound}$}
		\IF {$s.\text{count("0")}<s.\text{count("1")}$}
		\STATE reward$=1$
		\ELSE 
		\STATE reward$=0$
		\ENDIF

		\ELSIF {$s.\text{count('0')}>s.\text{count('1')}$}
		\STATE reward$=3$
        \ELSE 
        \STATE reward$=2$
		\ENDIF
	
	\end{algorithmic} 
\end{algorithm}

\subsection{Hyperparameters}
\begin{table}[h]
\centering
\caption{Hyperparameters}
\label{table:Hyperparameters}
\begin{tabular}{ll} 
\toprule
Hyperparameter                     & Value                                          \\ 
\hline
$\gamma$ (discount factor)         & 0.95                                           \\
$\mu$ (initial state distribution) & Uniform                                        \\
$\eta$ (learning rate)             & 0.1                                            \\
$\lambda$ (log barrier coefficient)                             & searched over $\{0.01, 0.1, 1.0, 10.0, 100.0\}$                                  \\
$K$ (NPG-BR inner-loop complexity)                                  & searched over $\{1,5,10,20,50\}$                                  \\
NN architecture                    & $2^N$-FC($2^N$)-FC($2^N$)-Linear($2$)-softmax  \\
\bottomrule
\end{tabular}

*The NN's input is the one-hot representation of the global state $s$.
\end{table}

\subsection{Computing resources}
The code is implemented by PyTorch, and a single run of 400 iterations took approximately 30, 50, 200 seconds for 2,3,5 agents version of the coordination game, respectively,   
using an NVIDIA Tesla V100 GPU and 32 CPU cores.

\newpage\clearpage
\section{Additional experimental results}
\subsection{Effect of $K$ for the NPG-BR dynamics}
\label{sec:Effect of $K$ for the NPG-BR dynamics}
In Figure \ref{fig:Softmax-POA-Nash-gap}, we plot the best-performing $K$ for $N=2,3,5$, respectively, in terms of the POA, with the results for each individual $K$ shown in Figure \ref{fig:effect_of_K}.
\begin{figure}[H]
    \centering
    \begin{subfigure}{\textwidth}
    \centering
    \includegraphics[width=.91\textwidth]{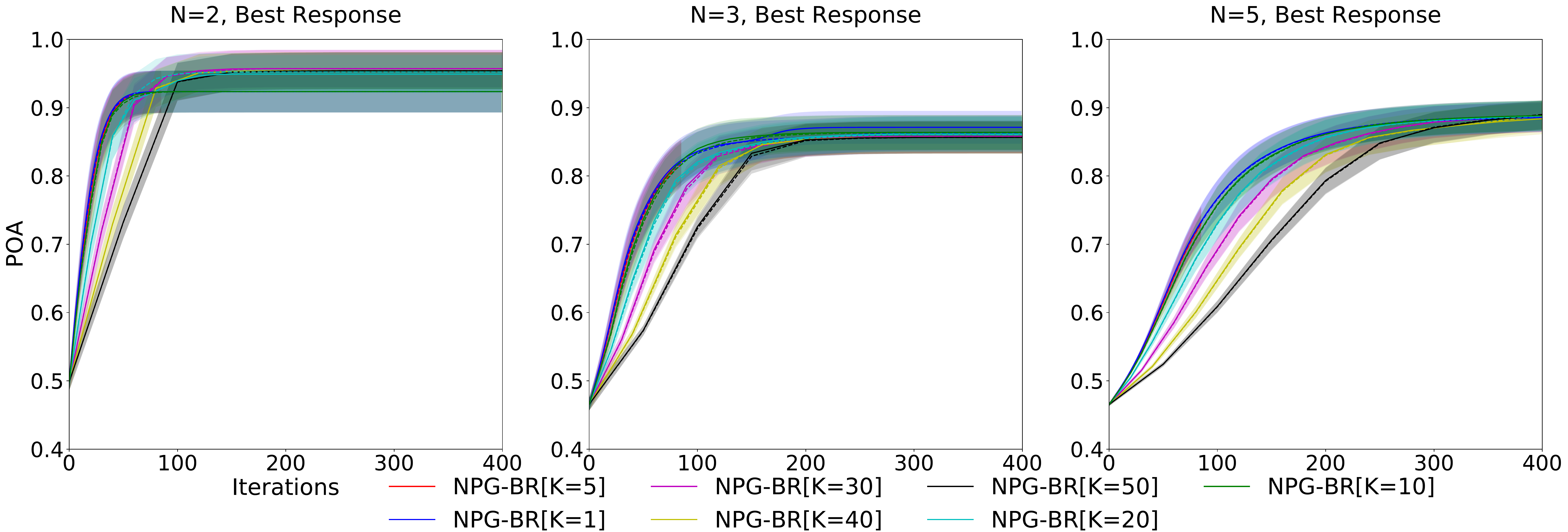}
    \end{subfigure}
    \begin{subfigure}{\textwidth}
    \centering
    \includegraphics[width=.91\textwidth]{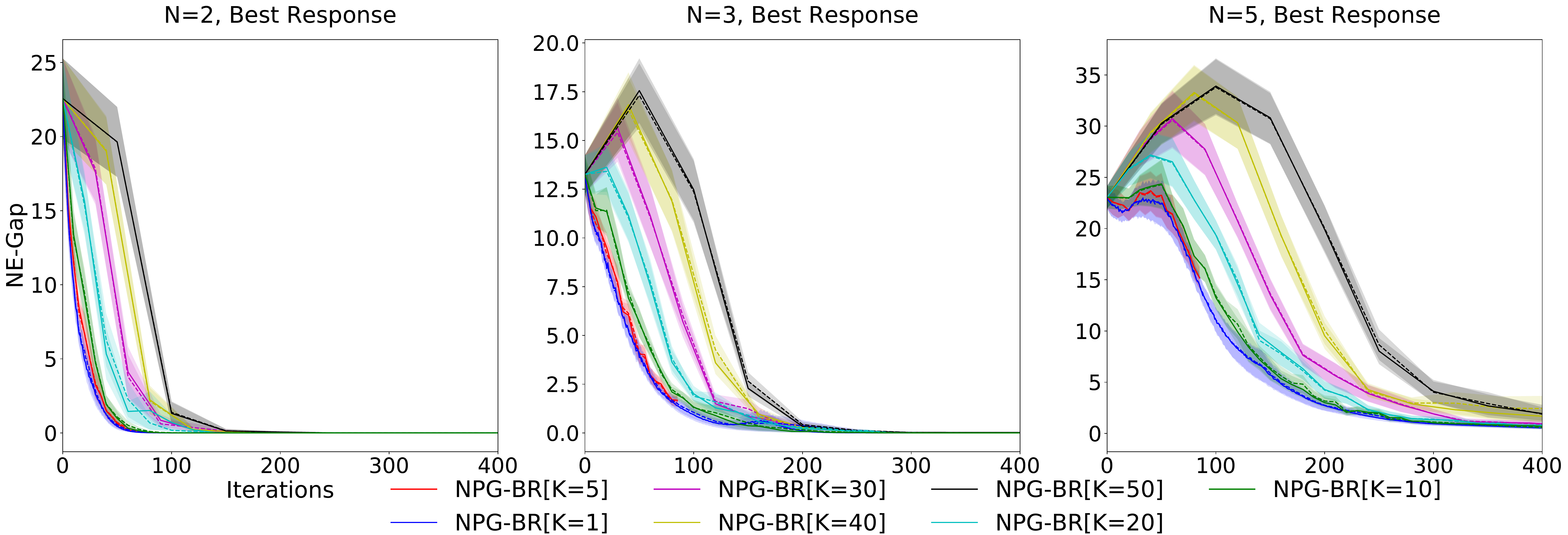}
    \end{subfigure}
    \caption{POA (top) and Nash-gap (bottom) under the tabular softmax parameterization (means and standard errors over 10 random initializations).The dashed lines are the curves of the log barrier regularized version of the algorithms with the same color. }
    \label{fig:effect_of_K}
\end{figure}

\subsection{Effect of the log barrier coefficient $\lambda$ for the PG dynamics under tabular softmax}
\begin{figure}[H]
    \centering
    \begin{subfigure}{\textwidth}
    \centering
    \includegraphics[width=.91\textwidth]{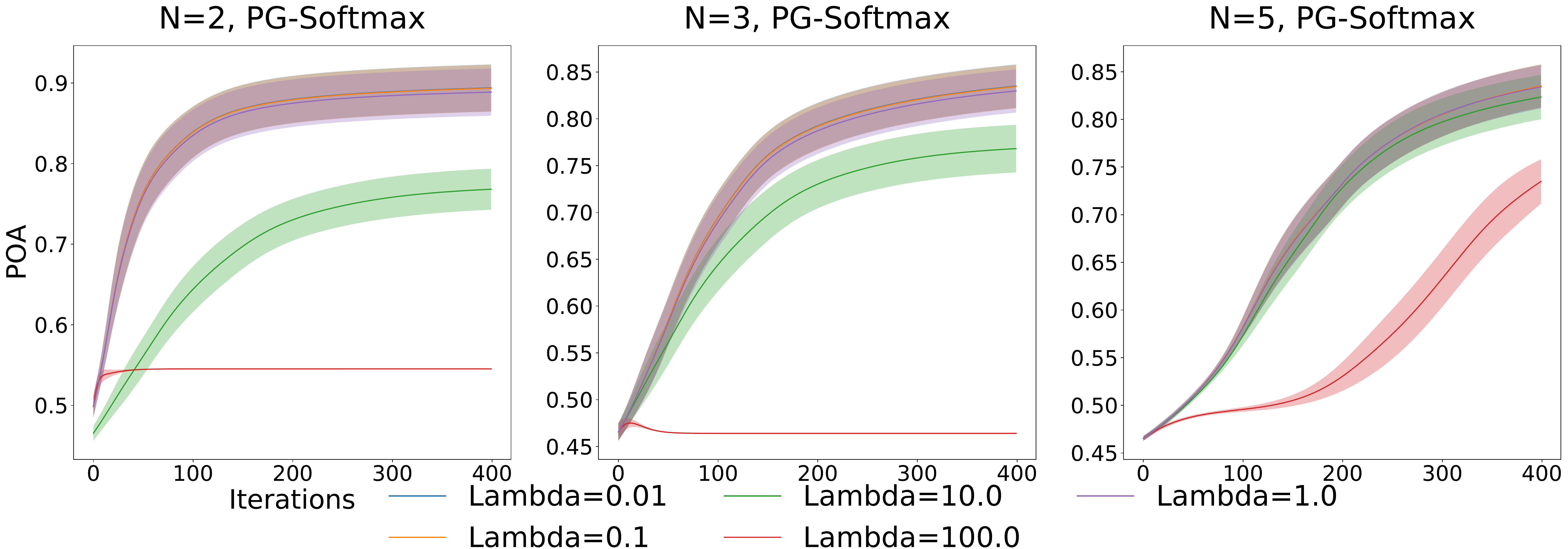}
    \end{subfigure}
    \begin{subfigure}{\textwidth}
    \centering
    \includegraphics[width=.91\textwidth]{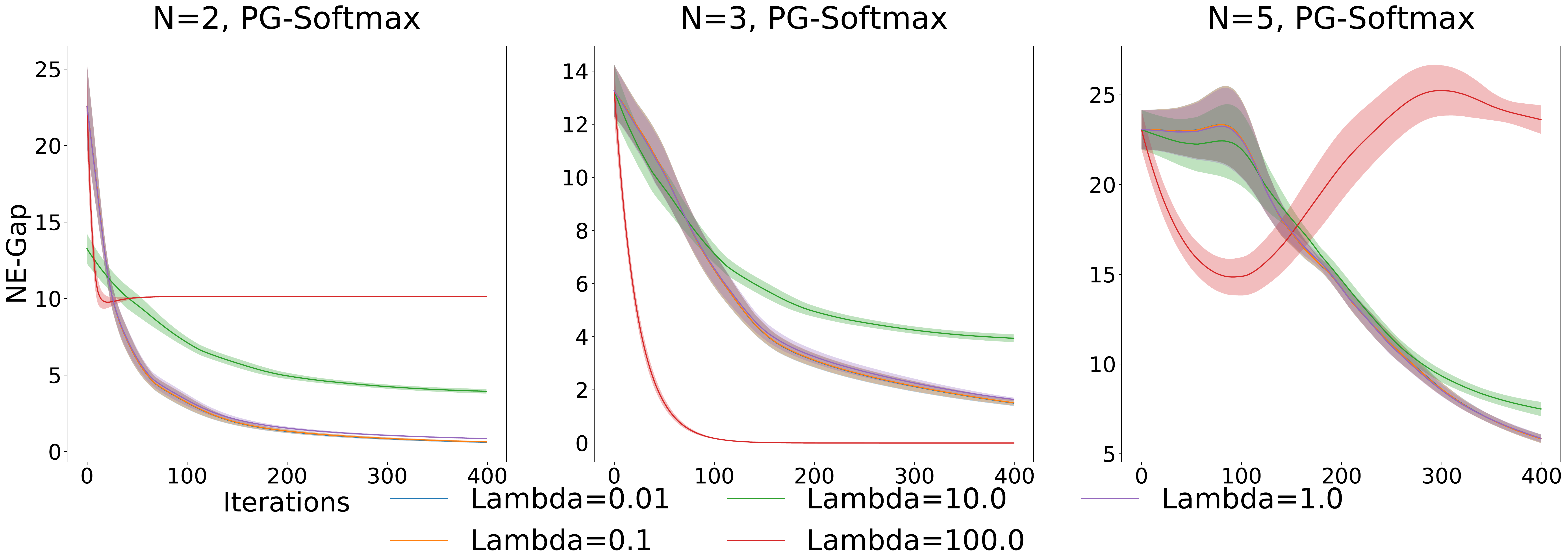}
    \end{subfigure}
    \caption{POA (top) and Nash-gap (bottom) for the PG dynamics under the tabular softmax parameterization (means and standard errors over 10 random initializations) with various choices for $\lambda$, the log barrier regularization coefficient. }
    \label{fig:effect_of_lambda}
\end{figure}
In Figure \ref{fig:Softmax-POA-Nash-gap}, we plot in the dashed lines the best-performing $\lambda$, the log barrier regularization coefficient, for the PG dynamics under tabular softmax.
Figure \ref{fig:effect_of_lambda} complement the results with the POA and Nash-gap curves with various choices for $\lambda$ we searched.

\end{document}